\documentclass[twocolumn,amsthm]{autart}    

\usepackage{amsmath,amssymb,amsfonts}
\usepackage{graphicx}
\usepackage{fancyvrb}
\usepackage{mathbbol}
\usepackage{ifthen}
\usepackage[]{color}
\usepackage[colorlinks, linkcolor=red, citecolor=green]{hyperref}
\usepackage{stackengine}

\newtheorem{thmm}{Theorem}

\newtheorem{assump}{Assumption}

\newtheorem{remark}{Remark}
\newtheorem{lemma}{Lemma}

\def\Busic{Bu\v{s}i\'{c}} 

\def\Prob{{\sf P}}

\def\stateSet{{\sf X}}
\def\stateSetZ{{\sf Z}}
\def\timeHorz{{\sf T}}
\def\delequalRHS{\mathrel{\ensurestackMath{\stackunder[1pt]{=}{\scriptstyle\triangledown}}}}

\def\textUB{\text{max}}
\def\textLB{\text{min}}

\def\textAgg{\text{\footnotesize{agg}}}
\def\textOn{\text{\footnotesize{on}}}
\def\textOff{\text{\footnotesize{off}}}
\def\textTS{\text{\footnotesize{TS}}}
\def\textBA{\text{BA}}
\def\textPolName{\text{GS}}

\def\tempBin{\ensuremath{\lambda}}
\def\numTCLs{\ensuremath{ {\sf{N_{tcl}} } }}

\def\Locked{\ensuremath{{\sf{L} } }}

\newcommand{\version}{arxiv}
\newcommand{\markedManu}{MARKED}

\def\pb#1{\footnote{{\color{red}{PB#1}}}}

\newboolean{showArxiv}
\newboolean{showArxivAlt}

\ifthenelse{\equal{\version}{journal}}
{
	\setboolean{showArxiv}{false} 
	\setboolean{showArxivAlt}{true}
}

\ifthenelse{\equal{\version}{arxiv}}
{
	\setboolean{showArxiv}{true} 
	\setboolean{showArxivAlt}{false}
}

\ifthenelse{\equal{\markedManu}{MARKED}}
{

}

\ifthenelse{\equal{\markedManu}{nonMARKED}}
{

}

\usepackage{xparse}

\ExplSyntaxOn
\cs_new:Nn \pb_prop_gset_bykeys:Nn
{
	\prop_clear:N \l__pb_temp_prop
	\keys_set:nn { pb/propbykey } { #2 }
	\prop_gset_eq:NN #1 \l__pb_temp_prop
}
\cs_generate_variant:Nn \pb_prop_gset_bykeys:Nn { c }

\keys_define:nn { pb/propbykey }
{
	unknown .code:n = \prop_put:Nxn \l__pb_temp_prop { \l_keys_key_tl } { #1 }
}

\cs_generate_variant:Nn \prop_put:Nnn { Nx }

\NewDocumentCommand{\setupcollaborator}{mm}
{
	\prop_new:c { g_collaborator_#1_prop }
	\pb_prop_gset_bykeys:cn { g_collaborator_#1_prop } { #2 }
}

\NewDocumentCommand{\selectcollaborator}{m}
{
	\prop_map_inline:cn { g_collaborator_#1_prop }
	{
		\tl_set:cn { ##1 } { ##2 }
	}
}
\ExplSyntaxOff

\setupcollaborator{PBmac}
{
	NAME=Prabir Mac , laptop or mac mini,
	DiCEbibPATH=/Users/pbarooah/Dropbox/Dropbox/DiCElab_bib,
	AustinbibPATH = /Users/pbarooah/Dropbox/Dropbox/austin_bib,
}
\setupcollaborator{Austin}
{
	NAME=Austin on linux,
	DiCEbibPATH=../../bibFiles/dicelab_bib,
	NRbibPATH = ../../../../Dropbox/austin_bib,
}

\selectcollaborator{Austin} 

\edef\endfrontmatter{%
	\unexpanded\expandafter{\endfrontmatter}
	\noexpand\endNoHyper 
}

\begin{document}





\begin{frontmatter}

\title{A unified framework for coordination of thermostatically controlled loads\thanksref{footnoteinfo}} 

\thanks[footnoteinfo]{This paper was not presented at any IFAC 
meeting. Corresponding author A.~Coffman. The research reported here has been partially supported by the NSF through awards 1646229 (CPS-ECCS) and 1934322 (CPS-ECCS), and the French National Research Agency grant ANR-16-CE05-0008.}

\author[AC]{Austin Coffman}\ead{bubbaroney@ufl.edu},    
\author[AB]{Ana \protect\Busic}\ead{ana.busic@inria.fr},               
\author[PB]{Prabir Barooah}\ead{pbarooah@ufl.edu}  

\address[AC]{University of Florida, Gainesville, FL, USA}  
\address[AB]{Inria, Paris, France}             
\address[PB]{University of Florida, Gainesville, FL, USA}        


\begin{abstract}
A collection of thermostatically controlled loads (TCLs) -- such as air conditioners and water heaters -- can vary their power consumption within limits to help the balancing authority of a power grid maintain demand supply balance. Doing so requires loads to coordinate their on/off decisions so that the aggregate power consumption profile tracks a grid-supplied reference. At the same time, each consumer's quality of service (QoS) must be maintained. While there is a large body of work on TCL coordination, there are several limitations. One is that they do not provide guarantees on the reference tracking performance and QoS maintenance.  A second limitation of past work is that they do not provide a means to compute a suitable reference signal for power demand of a collection of TCLs. In this work we provide a framework that addresses these weaknesses. The framework enables coordination of an arbitrary number of TCLs that: (i) is computationally efficient, (ii) is implementable at the TCLs with local feedback and low communication, and (iii) enables reference tracking by the collection while ensuring that temperature and cycling constraints are satisfied at every TCL at all times. The framework is based on a Markov model obtained by discretizing a pair of Fokker-Planck equations derived in earlier work by Malhame and Chong~\cite{MalhameElectricTAC:1985}. We then use this model to design randomized policies for TCLs. The balancing authority broadcasts the same policy to all TCLs, and each TCL implements this policy which requires only local measurement to make on/off decisions. Simulation results are provided to support these claims.
\end{abstract}
	
\end{frontmatter}

\ifshowArxiv
{
	\hypersetup{linkcolor=black}
	\tableofcontents
}
\fi

\section{Introduction}	
Many loads are flexible in their power demand: they can vary their demand around a baseline without adversely affecting consumers' quality of service (QoS). The flexibility can be used by a balancing authority (BA) to balance supply and demand in a power grid. The baseline demand refers to the power demand under normal operation, when each load operates only to meet its consumer's QoS without any interference from the BA. Since the rated power of each load is small, it is necessary to use a collection of loads. To provide grid support, the collection has to vary its demand from its baseline. It is envisioned that the BA would supply a reference signal for power demand and the actions of the loads in a collection would be coordinated so that their total demand tracks this reference. 

Thermostatically controlled loads (TCLs) - such as residential air conditioners, heat pumps, and water heaters - are recognized to be valuable sources of flexible demand~\cite{callaway2011achieving,ChenDistributedIMA:2017,matkoccal:2013,LeeGridJESBC:2020}. For an air conditioner or a heat pump, baseline demand is largely dictated by ambient weather conditions. There are at least two QoS requirements: the indoor temperature must be maintained within a prespecified range and compressor short-cycling must be avoided, meaning, once the compressor turns on it cannot turn off until a prespecified time period elapses, and vice versa. Coordination of TCLs involves two conflicting requirements: (i) the TCLs collectively need to track the reference power demand signal, and (ii) every TCL's QoS need to be maintained.

The actuation at each TCL is discrete: it can either be on or off. Direct load control~\cite{ChuNovelTPS:2008}, in which a centralized controller at the BA directly commands on/off status of each TCL is not scalable to large populations. A more scalable idea, that subsequent works on TCL coordination use, is for the BA to broadcast a low dimensional control command to all TCLs, which is translated by each TCL into its actuation command with a local policy. To avoid confusion between the decision making at the BA and a TCL, we use the word ``policy'' to mean the algorithm at a TCL that makes on/off decisions. The literature on decentralized coordination of TCLs differ in their choice of the broadcast signal (i.e., BA's control command) and the policy at the TCL that translates this broadcast to on/off decisions. Coordination architectures can be divided into two broad categories based on these choices: (i) thermostat set point change~\cite{callaway2011achieving,BashashModelingTCST:2013} and (ii) probabilistic control~\cite{matkoccal:2013,LiuDistributedTIE:2016,ChenDistributedIMA:2017,CoffmanVESBuildSys:2018}. These are discussed in more detail in Section~\ref{sec:lit}. 

A framework for coordinating TCLs needs two parts. The coordination scheme is one part. The other part is reference computation: the framework must provide the BA with a method to determine a suitable reference signal for the TCLs. That is, the reference must be such that the TCLs can collectively track the signal while each TCL maintains its QoS. Otherwise, even the best coordination scheme will fail to meet either the BA's need, which is reference tracking, or the consumers' need, which is maintaining indoor temperature etc., or both.

This work presents a unified framework for coordination of a collection of TCLs for providing grid support services. The framework enables both of the above mentioned components, i.e.,  (i) planning a suitable reference for a collection of TCLs and (ii) designing a randomized policy for coordination of the individual TCLs, so that both the BA's requirement and consumers' QoS are satisfied. In the proposed framework, the BA computes randomized control policies for the TCLs and broadcasts them to all the TCLs. Each TCL receives the same policy and implements it using locally measurable information. 
The framework is computationally tractable for an arbitrary number of TCLs. The communication burden is low: only a few numbers need to be broadcast by the BA at every sampling instant. Feedback from TCLs to the BA can be infrequent.

Underlying the framework is: (i) a Markov chain model that is derived from partial differential equations developed in the early work of Malhame and Chong~\cite{MalhameElectricTAC:1985}, (ii) state augmentation to incorporate cycling constraints, and (iii) convexification of the non-convex problem that appears in the design of the randomized control policy for the individual TCL. Additionally, we show that the assumption made about the effect of weather in earlier work~\cite{busmey:CDC:2016} on randomized control, under certain conditions, is in fact true. 

\subsection{Literature review and contribution}\label{sec:lit}

Before reviewing coordination methods, we discuss two interrelated modeling approaches that underpin many of the ideas in the TCL control architectures. These are the Markov chain and partial differential equation (PDE) models~\cite{KaraImpactTSG:2015,MalhameElectricTAC:1985,TotuDemandCST:2017,KhurramIdentificationEPSR:2020,zhang2013aggregated,NazirAnalysisEPSR:2020}, which stem from the early work of Malhame and Chong~\cite{MalhameElectricTAC:1985}. In~\cite{MalhameElectricTAC:1985} a pair of coupled Fokker-Planck equations are developed to model a collection of TCLs under thermostat control. The Fokker-Planck equations are PDEs that describe the time evolution of a certain probability density functions (pdf) over the state space of temperature and on/off mode. The PDEs can be used to model the entire collection or a single TCL: the probability that a single TCL is ``on'' is approximately the fraction of TCLs that are ``on''. Discretizing the PDE yields a Markov chain model, though some works have obtained Markov models without using the PDEs. Hence, \emph{one} set of PDEs can model a collection of TCLs. Thus, methods that base control design on the PDE or Markov chain framework scales well with the number of TCLs.

Due to the lack of scalability of direct load control, we limit our attention to the two broad classes mentioned earlier: (i) thermostat set point changes, (ii) probabilistic policy. There are many forms of probabilistic policy, which can be roughly subdivided into two sub categories: (ii-A) ``bin switching'' and (ii-B) ``randomized policy''. We discuss these in detail below.

In the thermostat setpoint change coordination architecture, a time-varying thermostat set point is broadcast to all TCLs, and each TCL makes on/off decisions based on this new setpoint~\cite{callaway2011achieving,BashashModelingTCST:2013}. This approach may ask for an extremely small change in thermostat setpoint, far below the resolution of the temperature sensor at each TCL. Or it may ask for large changes in thermostat setpoint which will violate occupant comfort. 

In a probabilistic policy architecture, the TCL policy - the mapping from BA's broadcast command to a TCL's on/off decision -  is a non-deterministic mapping. Works in this category typically first model the population of TCLs under thermostat control, which is a deterministic policy, as a Markov chain. The continuous temperature range is divided into a number of discrete bins. A finite dimensional state vector, a probability mass function, is then defined. Each entry of the state vector represents ``the fraction of TCLs that are on (or off) and has temperature in a certain range.'' 

Since the basic Markov model is derived for the thermostat policy, introduction of the BA's control to manipulate TCLs' on/off state is somewhat ad-hoc. In the the bin switching literature, the control command from the BA is chosen so as to affect the fraction of TCLs in the temperature bins directly.  In~\cite{matkoccal:2013}, the BA's control command is chosen to be another vector, whose $i^{\text{th}}$ entry represents ``the fraction of TCLs in bin $i$ to increase/decrease''. A policy is then proposed to translate this command to on/off action at each TCL, which requires knowledge of the state of the Markov model. In~\cite{LiuDistributedTIE:2016}, BA's control command is chosen to be a scalar. The probability of a TCL turning on or off is proportional to this scalar. Subsequent works have proposed various refinements, such as BA's command affecting the rate of fractions to switch instead of fraction to switch~\cite{TotuDemandCST:2017}. Providing performance guarantees with bin switching architecture has proved challenging, either on reference tracking or on QoS maintenance for individual TCLs. 


An alternative to bin switching that still uses probabilistic on/off decision making is randomized policy~\cite{busmey:CDC:2016,ChenDistributedIMA:2017}. A randomized policy is a specification of the conditional probability of turning on or off given the current state. On/off decisions are computed with the help of a random number generator and the policy. In this architecture it is envisioned that the thermostat policy at the TCL is replaced with a randomized policy. In \cite{busmey:CDC:2016,ChenDistributedIMA:2017}, the policy is parameterized by a scalar $\zeta(t)$. Coordination of the population is then achieved by appropriate design of $\zeta(t)$, which is computed and broadcast by the BA. This architecture also uses a Markov model of the evolution of binned temperature, but assumes a certain factorization: the next values of the temperature and mode are conditionally independent given the current joint pair of temperature and mode values under the effects of the randomized policy and exogenous disturbances, especially weather. That is, the transition matrix of the state process is a point wise product of two controlled transition matrices. In an optimal control setting, computation of the BA's control command, $\zeta(t)$, for reference tracking is a non-convex optimization problem~\cite{CoffmanAggregateCDC:2019}. The probability of turning on when temperature exceeds the upper limit, or off when temperature dips below the lower limit, is set to 1 by design. This will ensure the temperature constraint is maintained. Attempts have been made to maintain the cycling constraint~\cite{CoffmanVESBuildSys:2018}. But a formal design method to incorporate the cycling constraint has been lacking.  

A complete framework for coordination of TCL needs not only a control algorithm to make decisions at TCLs, but also a method to compute a \emph{feasible} reference signal for the collection's power demand. Feasible means that no TCL needs to violate local constraints in order for the collection to track the reference. The topic is sometimes described as ``flexibility capacity'' and has been examined in many recent works, with various definitions of flexibility~\cite{hao_aggregate:2015,PaccagnanRangeCDC:2015,CoffmanCharacterizingTPS:2020,coffmanFlexibilityTCL_ArXiV:2020}. A unified treatment of reference design and coordination algorithm design that would provide a complete framework is lacking. 

In short, existing work on TCL coordination has a number of scattered disadvantages. Direct load control suffers from scalability/privacy issues and thermostat set-point methods have implementation issues. Bin switching does not provide guarantees on reference tracking and often requires solving a challenging state estimation problem. Prior work on randomized control requires non-convex optimization and is based on an assumed conditional independence. Finally, there is a lack of unified treatment of the reference design and policy design problems.

In this work we develop a unified framework for coordination of TCLs that addresses the weaknesses of prior work described above. Our major contributions are as follows.
\begin{enumerate}
  \item We provide a complete framework that allows the  BA to compute (a) an optimal reference signal that is feasible for the collection and (b) optimal randomized policies for the TCLs. When the TCLs implement these policies, their total power demand collectively tracks the reference signal and the policies guarantee that temperature and cycling QoS requirements at each TCL are satisfied. Optimal reference means it is closest to what the BA wants while being feasible for the TCLs. Implementation of the policy at a TCL is easy; it requires only local measurements. The communication burden for coordination is also low.  At each sampling time, a randomized control policy - parameterized by a few numbers -  is broadcast to all TCLs. Feedback from TCLs to the BA can be infrequent. 
\item Our framework is based on a careful discretization of the partial differential equation (PDE) model described in~\cite{MalhameElectricTAC:1985}. This discretization shows that a certain ``conditional independence'' that was assumed in~\cite{busmey:CDC:2016} indeed holds. The conditional independence separates the effects of the policy at the TCL (control) and weather (disturbance) on the transition matrix, and greatly facilitates computation of policies.
\item  Numerical experiments are provided to illustrate the efficacy of the framework. Simulations show that TCLs are able to track the optimal reference collectively while each TCL is able to maintain both temperature and cycling constraints. Matlab implementation is made publicly available at~\cite{CoffmanUnifiedCode:2021}.
\end{enumerate} 
Figure~\ref{fig:controlArch} illustrates the two parts of the proposed framework.

\begin{figure}
	\centering
	\includegraphics[width=1\columnwidth]{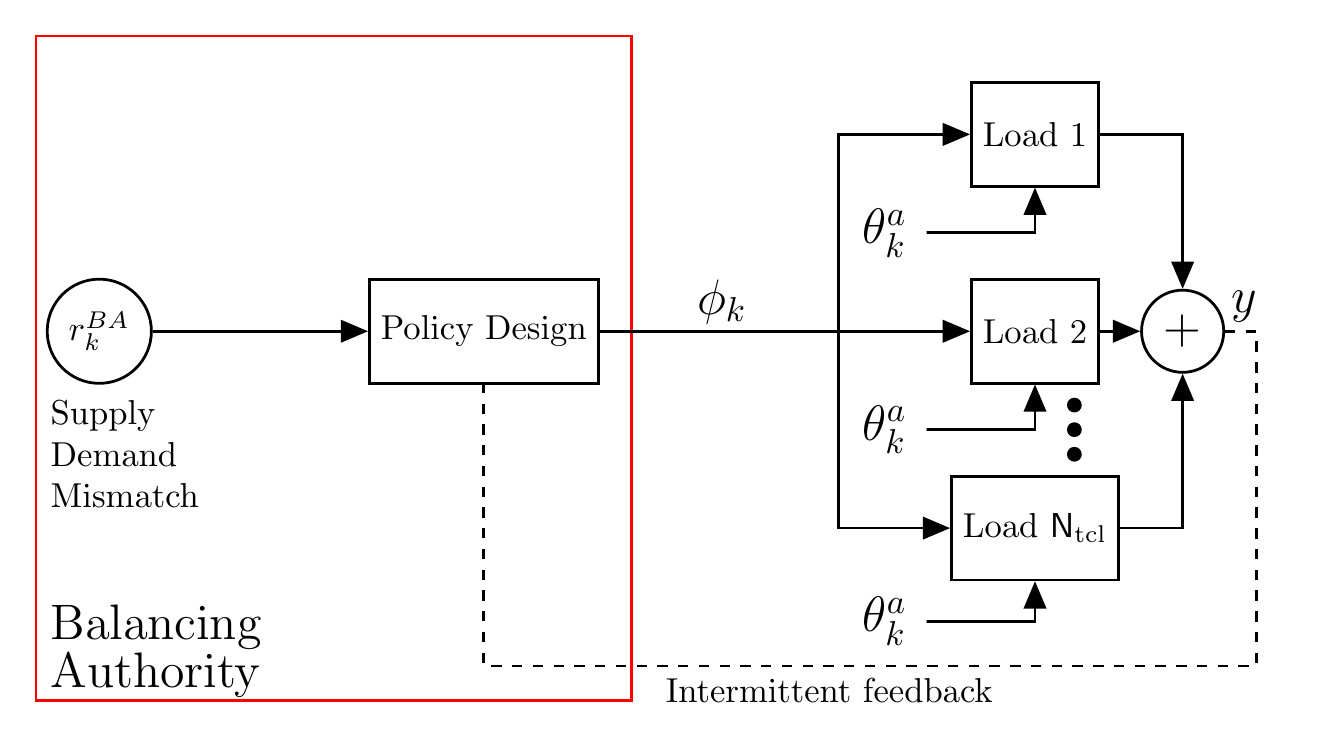}
	\caption{Coordination architecture with the proposed framework.}
	\label{fig:controlArch}
\end{figure}


\ifshowArxiv
The Markov model obtained by discretizing a PDE was presented in~\cite{CoffmanControlACC:2021}. For completeness, we include the discretization in this paper as an Appendix. 
\else
The rest of the paper proceeds as follows. In Section~\ref{sec:ModelDev} the model of the individual TCL is introduced. In Section~\ref{sec:discretzation} the PDEs introduced are discretized and in Section~\ref{sec:indStruc} the structure of the discretized model is identified. Since the PDE discretization was previously reported in~\cite{CoffmanControlACC:2021}, technical details including some of the proofs are moved to the expanded version~\cite{CoffmanUnifiedArxiv:2021}; only the parts necessary for completeness are described in this section. The proposed framework for reference and policy design is presented in Section~\ref{sec:propFrame} and numerical experiments are reported in Section~\ref{sec:numExp}. 
\fi

\subsection{Notation}
The symbol $\mathbb{1}$ denotes the vector of all ones, $\mathbf{e}_i$ denotes the i$^{th}$ canonical basis vector, and $\mathbf{0}$ denotes the zero matrix or vector, all of appropriate dimension. For a vector $v$,  $\text{diag}(v)$ denotes the diagonal matrix with entries of $v$, i.e., $\text{diag}(v)\mathbb{1} = v$. Further, $\otimes$ denotes matrix Kronecker product and $\mathbf{I}_A(\cdot)$ the indicator function of the set $A$.
\section{Modeling: Individual TCL} \label{sec:ModelDev}

A thermostatically controlled load (TCL) is an on/off device that ensures the temperature of a given environment remains within a specified region, e.g., an air conditioner.
During its operation, the TCL must adhere to certain operational requirements (QoS constraints). We consider two: the temperature constraint and the cycling constraint. The temperature constraint is that the TCL's temperature must remain within a prespecified deadband, $[\lambda^{\min}, \lambda^{\max}]$. This is achieved by switching the TCL on or off when it is too hot or cold. The cycling constraint is that the TCL can only change from ``on'' to ``off'' or vice versa once every $\tau$ (discrete) time instants, where $\tau$ is a prespecified constant. The cycling constraint is to ensure the mechanical hardware is not damaged. In both cases, ensuring the two constraints amounts to appropriately deciding when to switch the TCL on or off. 
\ifx 0
\begin{assump} \label{assump:cycConstDecup}
	Let $x_m(t;x_0)$ denote the temperature of the TCL with mode $m$ after $t$ time while starting at $x_0$. Then define the following times, $t^*_1$ and $t^*_2$ as
	\begin{align} \nonumber
		x_{\textOff}\big(t^*_1; \ x_{\textOn}(t_0 + \tau;\tempBin^{\textUB})\big) = \tempBin^{\textUB}, \quad \text{and} \quad x_{\textOn}\big(t^*_2; \ x_{\textOff}(t_0 + \tau;\tempBin^{\textLB})\big) = \tempBin^{\textLB},  
	\end{align}
	where $t_0$ is an arbitrary initial time. We require $\tau \leq t^* \triangleq \min\{t^*_1, \ t^*_2 \}$ and that 
	\begin{align}
		s^{\textOn}(t_0) &= \begin{cases}
		1, & x_{off}(t_0+\tau;x_0) > \tempBin^{\textUB} \ \text{and} \ m(t_0) = 1. \\
		0, & \text{otherwise}
		\end{cases} \\
		s^{\textOff}(t_0) &= \begin{cases}
		1, & x_{on}(t_0+\tau;x_0) < \tempBin^{\textLB} \ \text{and} \ m(t_0) = 0. \\
		0, & \text{otherwise}
		\end{cases}
	\end{align} 
\end{assump}
\fi
\subsection{Temperature dynamics of TCLs}
The typical model for the TCL's temperature $\theta(t)$ in the literature is the following ordinary differential equation (ODE),
\begin{align} \label{eq:detModelTCL}
  \begin{split}
	\frac{d}{dt} \theta(t) &= f_m(\theta,t),
\quad \text{with} \\
	f_m(\theta,t) &= -\frac{1}{RC}\left(\theta - \theta^a(t)\right) - m(t)\frac{\eta P_0}{C}.    
  \end{split}
\end{align}
The rated electrical power consumption is denoted $P_0$ with coefficient of performance (COP) $\eta$. The parameters $R$ and $C$ denote thermal resistance and capacitance, respectively. The signal $\theta^a(t)$ is the ambient temperature. The quantity $m(t)$ is the on/off mode, and in the following we identify $m(t) = 1$ and $m(t) =$ on, as well as $m(t) = 0$ and $m(t) =$ off. We denote arbitrary temperature values through the variable $\lambda$, and the thermostat setpoint as $\lambda^{\text{set}}$. The values $\lambda^\textUB$ and $\lambda^\textLB$ set the upper and lower limit for the temperature deadband.

A model for the temperature state that accounts for modeling errors in~\eqref{eq:detModelTCL} and will be crucial in developing the content in Section~\ref{sec:pdeModel} is the following It\^{o} stochastic differential equation (SDE),
\begin{align} \label{eq:stoModelTCL}
d \theta(t) = f_m(\theta,t)dt + \sigma dB(t).
\end{align}
The term $B(t)$ is Brownian motion with parameter $\sigma>0$, and the quantity $\sigma dB(t)$ captures modeling errors in~\eqref{eq:detModelTCL}. 
In either model, the baseline power for the TCL is the value of $P$ so that $f_1(\lambda^{\text{set}},t) = 0$, solving yields:
\begin{align} \label{eq:invTCLbase}
\text{Baseline Power:} \quad \bar{P}^{\text{ind}}(t) = \frac{\theta^a(t)-\lambda^{\text{set}}}{\eta R}.
\end{align} 
For $\numTCLs$ TCLs the baseline power $\bar{P}(t)$ and maximum power $P_\textAgg$ are, 
\begin{align} \label{eq:powerAgg}
\bar{P}(t) \triangleq \numTCLs\bar{P}^{\text{ind}}(t), \quad \text{and} \quad P_\textAgg \triangleq \numTCLs P_0.
\end{align}
The total electrical power consumption of the collection, whether with thermostat policy or some other policy, is denoted by $y(t)$:
\begin{align}
  \label{eq:y-def}
  y(t) \triangleq P_0\sum_{\ell=1}^{\numTCLs}m^\ell(t) 
\end{align}
where $m^\ell(t)$ is the on/off state of the $\ell$-th TCL.

\subsubsection{Policy (at the TCL)}
The mode state of a TCL evolves according to a policy. The following policy, which we denote as the \emph{thermostat policy}, ensures the temperature constraint:
\begin{align} \label{eq:thermoContLaw}
\lim_{\epsilon \rightarrow 0}\ m(t + \epsilon) = \begin{cases}
1, & \theta(t) \geq \tempBin^{\textUB}. \\
0, & \theta(t) \leq \tempBin^{\textLB}. \\
m(t), & \text{o.w.} 
\end{cases}
\end{align}

We add the following set of assumptions about the individual TCL discussed so far.

\begin{enumerate} 
	\item[\textbf{A.1}] The thermostat policy does not violate the cycling constraint. 
	\item[\textbf{A.2}] For all $t\geq 0$ and $\theta\in [\lambda^\textLB,\lambda^\textUB]$, $f_\textOn(\theta,t) \leq 0$ and $f_\textOff(\theta,t) \geq 0$.
	\item[\textbf{A.3}] The TCL's cycling and temperature constraint are both simultaneously feasible.
\end{enumerate}

The sizing/design of the TCL is most likely to ensure that \textbf{A.1}  holds. With \textbf{A.1} , we depart from discussing the cycling constraint until Section~\ref{sec:propFrame} since up to that point the mode state is assumed to evolve according to~\eqref{eq:thermoContLaw}. 

Assumption \textbf{A.2} states that when the TCL is on, the temperature does not increase and when the TCL is off the temperature does not decrease. All prior works focusing on cooling TCLs (e.g., air conditioners) implicitly make this assumption. Every result that is to follow is also valid for heating TCLs (e.g., a water heater or a heat pump) with a sign reversal.

Like \textbf{A.2}, assumption \textbf{A.3} is also implicit in any work that considers both the TCLs temperature and cycling constraint.

\subsection{PDE model} \label{sec:pdeModel}
We now describe a PDE model of a TCL's temperature with thermostat policy originally derived in~\cite{MalhameElectricTAC:1985}. Consider the following marginal pdfs $\mu_{\textOn}, \mu_{\textOff}$:
\begin{align} \label{eq:probOnState}
	\mu_{\textOn}(\tempBin,t)d\tempBin &= \Prob\left((\tempBin < \theta(t) \leq \tempBin + d\tempBin), \ m(t) = \text{on} \right), \\ \label{eq:probOffState}
	\mu_{\textOff}(\tempBin,t)d\tempBin &= \Prob\left((\tempBin < \theta(t) \leq \tempBin + d\tempBin), \ m(t) = \text{off} \right),
\end{align}
where $\Prob(\cdot)$ denotes probability, $\theta(t)$ evolves according to~\eqref{eq:stoModelTCL} and for now $m(t)$ evolves according to~\eqref{eq:thermoContLaw}. It was shown in~\cite{MalhameElectricTAC:1985} that the densities $\mu_{\textOn}$ and $\mu_{\textOff}$ satisfy the Fokker-Planck equations,
\begin{align} 
	\frac{\partial}{\partial t}\mu_{\textOn}(\lambda,t) &= \frac{\sigma^2}{2}\nabla^2_\lambda\mu_{\textOn}(\lambda,t)-\nabla_\lambda\Big(f_{\textOn}(\lambda,t)\mu_{\textOn}(\lambda,t)\Big)  \label{eq:pdeOnMode} \\
	\frac{\partial}{\partial t}\mu_{\textOff}(\lambda,t) &= \frac{\sigma^2}{2}\nabla^2_\lambda\mu_{\textOff}(\lambda,t)-\nabla_\lambda\big(f_{\textOff}(\lambda,t)\mu_{\textOff}(\lambda,t)\big) \label{eq:pdeOffMode}
\end{align}
that are coupled through their boundary conditions~\cite{MalhameElectricTAC:1985}. \ifshowArxiv The boundary conditions are listed in Appendix~\ref{app:boundCondProof}. \else The boundary conditions are listed in~\cite{CoffmanUnifiedArxiv:2021}. \fi

\begin{remark} \label{rem:LLN}
The coupled equations~\eqref{eq:pdeOnMode}-\eqref{eq:pdeOffMode} can be used to model either: (i) a \emph{single} TCL or (ii) a \emph{collection} of TCLs. For (i) the quantities~\eqref{eq:probOnState}-\eqref{eq:probOffState} represent the \emph{probability} that a single TCLs temperature and on/off mode reside in the respective region. For (ii) the quantities~\eqref{eq:probOnState}-\eqref{eq:probOffState} represent the \emph{fraction} of TCLs whose temperature and on/off mode reside in the respective region. How the equations~\eqref{eq:pdeOnMode}-\eqref{eq:pdeOffMode} (specifically their discretized form) can be used to model an ensemble is discussed further in Section~\ref{sec:aggModel}. 
\end{remark}

\ifx 0  
\begin{figure}
	\centering
	\includegraphics[width=0.75\columnwidth]{histSimCompare_advectEq.pdf}
	\caption{Discrepancy between the state of the advection equation and the histogram of the population for various time steps. Each histogram is over the temperature state for all of the on TCLs at the specified time.}
	\label{fig:advEqPeriodic}
\end{figure}
\subsubsection{Motivation for Stochastic Model}
While transport type arguments can be used to develop a pair of coupled advection equations (equations~\eqref{eq:pdeOnMode}-\eqref{eq:pdeOffMode} with $\sigma^2=0$) for the deterministic model~\cite{BashashModelingTCST:2013}, the state of these advection equations will not agree with the pointwise in time histogram of a population of TCLs simulated with~\eqref{eq:detModelTCL} (see Figure~\ref{fig:advEqPeriodic}). To see why, consider the following: without noise TCLs are periodic whereas discretization of the advection equations yields a Markov transition matrix that is irreducible and aperiodic. Hence, the iteration of this transition matrix will converge to a limiting and invariant distribution, whereas the samples from the TCLs will not since the TCL behavior is periodic. This behavior is shown in Figure~\ref{fig:advEqPeriodic}, the discretized state of the advection equation remains relatively constant while the histogram of the collection does not; their is no suggestion of convergence even after 24 hours.

Thus, from the previous argument, the stochastic model~\eqref{eq:stoModelTCL} has two advantages: (i) it captures modeling errors and heterogeneity~\cite{MouraModelingDSCC:2013} the deterministic model~\eqref{eq:detModelTCL} can not, and (ii) it also guarantees a correspondence between simulation samples from~\eqref{eq:stoModelTCL} and the state of the coupled PDEs~\eqref{eq:pdeOnMode}-\eqref{eq:pdeOffMode} (see Figure~\ref{fig:histSimCompare}). 
\fi

\ifx 0
\subsubsection{Forward thinking motivation}Further, the PDEs that are derived from the stochastic model will be the base of our control oriented model. In the following, we will discretize the PDEs~\eqref{eq:pdeOnMode}-\eqref{eq:pdeOffMode} and then show that the discretized model has special structure. Particularly, the structure elucidates how to model the aggregate under the effects of a arbitrary randomized policy. 
\fi
\begin{figure*}
	\centering
	\includegraphics[width=1.75\columnwidth]{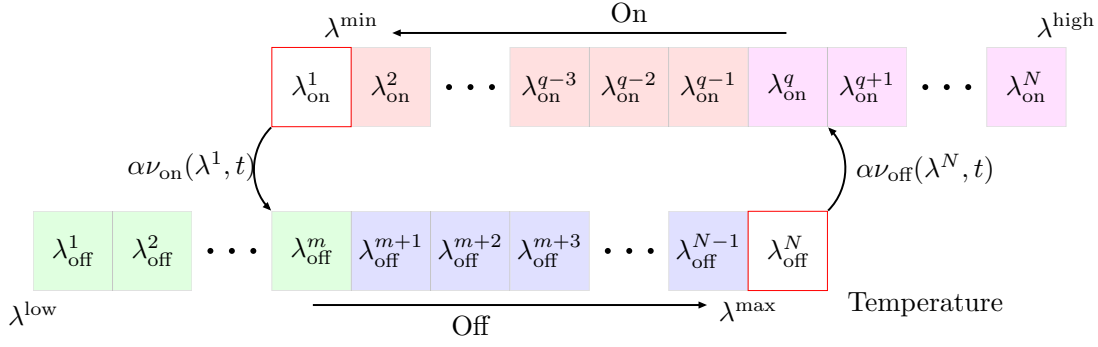}
	\caption{The control volumes (CVs). The colors correspond to the colors found in Figure~\ref{fig:sparPattern}. The values in each CV represent the nodal temperature for the CV. The arrows describe the sign of the convection of the TCL through the CVs. The values are such that $N = m + q$. The terms involving $\alpha$ model rate of transfer between the corresponding CVs due to the thermostat policy, where $\alpha=\gamma + \frac{\sigma^2}{(\Delta \lambda)^2}$. The parameter $\gamma>0$ is a design parameter; see Remark~\ref{rem:conditionsForCondFact}.}
	\label{fig:cvLayout}
\end{figure*}

\section{Markov model from PDE Discretization} \label{sec:discretzation}
We use the finite volume method (FVM) to discretize the PDEs~\eqref{eq:pdeOnMode} and~\eqref{eq:pdeOffMode}. The discretization of~\eqref{eq:pdeOnMode} and~\eqref{eq:pdeOffMode} yields a finite dimensional probabilistic model for a single TCL (equation~\eqref{eq:discDynEsem}). We discretize the PDEs~\eqref{eq:pdeOnMode} and~\eqref{eq:pdeOffMode} in a way that a control input for the BA can then be identified. More on this point will be discussed in Section~\ref{sec:indStruc}, however the discretization here will play a role. 


\subsection{Spatial discretization}
The FVM bins the continuous temperature into $N$ control volumes (CV). The layout of the CVs is shown in Figure~\ref{fig:cvLayout}. The $N$ CVs for both the on and off mode state, as shown in Figure~\ref{fig:cvLayout}, are defined through the nodal temperature values ($\lambda_{\textOn}$ and $\lambda_{\textOff}$) and their boundaries ($\lambda^+_{\textOn}$ and $\lambda^+_{\textOff}$) and  ($\lambda^-_{\textOn}$ and $\lambda^-_{\textOff}$):
\begin{align} \nonumber
	&\lambda_{\textOn} = (\lambda^{i}_{\textOn})_{i=1}^{N}, \quad  \lambda_{\textOn}^+ = \lambda_{\textOn} + \frac{\Delta\lambda}{2}, \quad \lambda_{\textOn}^- = \lambda_{\textOn} - \frac{\Delta\lambda}{2}, \\ \nonumber
	&\lambda_{\textOff} = (\lambda^{i}_{\textOff})_{i=1}^{N}, \quad \lambda_{\textOff}^+ = \lambda_{\textOff} + \frac{\Delta\lambda}{2}, \quad \lambda_{\textOff}^- = \lambda_{\textOff} - \frac{\Delta\lambda}{2},
\end{align}
where $\Delta\lambda$ is the CV width.
All intermediate values of $\lambda_{\textOn}$ and $\lambda_{\textOff}$ are separated from each other by $\Delta\lambda$. The values in $\lambda^+_{\textOn}$ (respectively, $\lambda^+_{\textOff}$) are the right edges of the CVs and the values $\lambda^-_{\textOn}$ (respectively, $\lambda^-_{\textOff}$) are the left edges of the CVs, for example, $\lambda^{1,-}_{\textOff} = \lambda^{\text{low}}$. The quantities $\lambda^{\text{min}}$ and $\lambda^{\text{max}}$ specify the thermostat deadband, and are \emph{different} from the quantities $\lambda^{\text{high}}$ and $\lambda^{\text{low}}$ (see Figure~\ref{fig:cvLayout}).

\begin{figure}
	\centering
	\includegraphics[width=1\columnwidth]{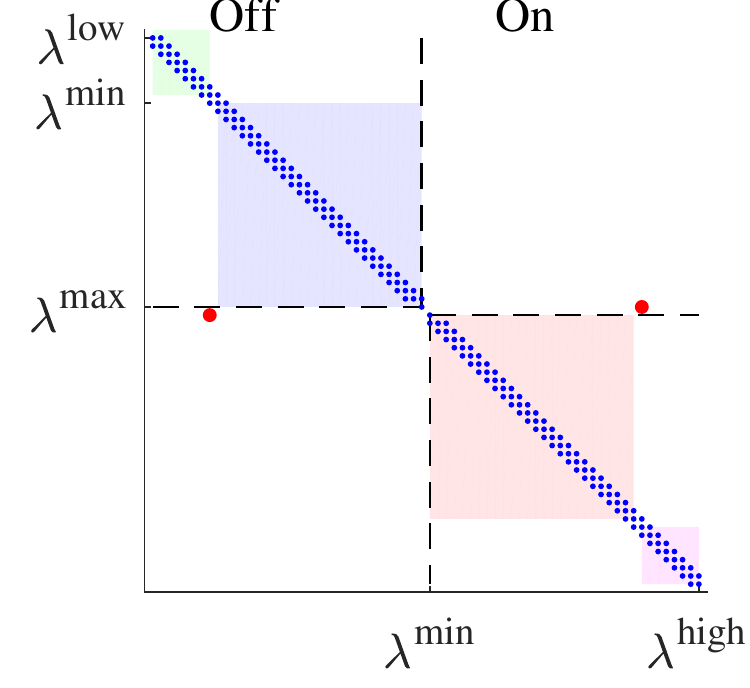}
	\caption{Sparsity pattern of the matrix $A(t)$ for $N=51$ CVs for both the on and off state. The colors correspond to the colors found in Figure~\ref{fig:cvLayout}.}
	\label{fig:sparPattern}
      \end{figure}

\ifshowArxiv      
The steps taken to obtain the spatially discretized PDEs is detailed in  Appendix~\ref{app:pdeDisc}. To give an overview, the discretization is done in two parts: (i) for the internal CV's (Appendix~\ref{app:pdeDiscInt}) and (ii) for the boundary CV's (Appendix~\ref{app:boundCondProof}). We describe here the end result of the derivation in Appendix~\ref{app:pdeDisc}.
\else
The steps taken to obtain the spatially discretized PDEs is detailed in~\cite{CoffmanUnifiedArxiv:2021}. To give an overview, the discretization is done in two parts: (i) for the internal CV's and (ii) for the boundary CV's. We describe here the end result of the derivation.
\fi
First, define the following quantities
\begin{align}
	\nu_{\textOff}(\lambda^i,t) &\triangleq \mu_{\textOff}(\lambda^i,t)\Delta\lambda, \quad \text{and} \\ \nu_{\textOn}(\lambda^i,t) &\triangleq \mu_{\textOn}(\lambda^i,t)\Delta\lambda,
\end{align}
then construct the row vector, $\nu(t) = [\nu_{\textOff}(t),\nu_{\textOn}(t)]$. with
\begin{align}
	\nu_{\textOff}(t) &\triangleq [\nu_{\textOff}(\lambda^1,t), \dots,\nu_{\textOff}(\lambda^N,t)], \quad \text{and} \\
	\nu_{\textOn}(t) &\triangleq [\nu_{\textOn}(\lambda^1,t), \dots,\nu_{\textOn}(\lambda^N,t)].
\end{align}
By combining all the ordinary differential equations (ODEs) for the $\nu_{\textOff}(\lambda^i,t), \nu_{\textOn}(\lambda^i,t)$ for all the $i$'s, we obtain the linear time varying system
\begin{align} \label{eq:dynContMC}
\frac{d}{dt}\nu(t) = \nu(t)A(t).
\end{align}
The sparsity pattern of $A(t)$ is shown in Figure~\ref{fig:sparPattern}. The system~\eqref{eq:dynContMC} is the spatially discretized version of the PDEs~\eqref{eq:pdeOnMode}-\eqref{eq:pdeOffMode}. 
\ifx 0
We now define certain sub blocks of the matrix $A(t)$, which are important for later discussion. Temporarily introducing the notation $A_{[i:j,l:k]}$ to denote the submatrix of $A$ with rows $i$ to $j$ and columns $l$ to $k$, we then define
\begin{align} \label{eq:AonOff}
	A^\textOn(t)&\triangleq A(t)_{[N+1:2N,N+1:2N]}, \quad \text{and} \\
	A^\textOff(t)&\triangleq A(t)_{[1:N,1:N]}.
\end{align}
Using the definitions of $A^\textOn(t)$ and $A^\textOff(t)$ we also define,
\begin{align} \label{eq:AHatonOff}
\hat{A}^\textOn(t)&\triangleq A^\textOn(t)_{[1:q-1,1:q-1]}, \quad \text{and} \\ \nonumber 
\hat{A}^\textOff(t)&\triangleq A^\textOff(t)_{[m+1:N,m+1:N]},
\end{align}
where $m$ and $q$ denote certain CV indices (shown in Figure~\ref{fig:cvLayout}).
The above can all be visualized through the shaded regions ($\hat{A}^\textOn(t)$ and $\hat{A}^\textOff(t)$) and the regions separated by the black dashed lines ($A^\textOn(t)$ and $A^\textOff(t)$) in Figure~\ref{fig:sparPattern}. \fi
The matrix $A(t)$ also satisfies the properties of a transition rate matrix, described in the following lemma.
\begin{lemma} \label{lem:rateMat}
	For all $t$, the matrix $A(t)$ is a transition rate matrix. That is, for all $t$
	\begin{align} \nonumber
	\text{(i):} &\quad A(t)\mathbb{1} = \mathbf{0}. \\ \nonumber
	\text{(ii):} &\quad \text{for all} \ i, \ A_{i,i}(t) \leq 0, \ \text{and} \ \text{for all} \ j\neq i \  A_{i,j}(t) \geq 0. 
	\end{align}
\end{lemma}
\ifshowArxiv
\begin{proof}
	See Appendix~\ref{app:rateMatProof}.
      \end{proof}
      \else
\begin{proof}
See \cite{CoffmanUnifiedArxiv:2021}.
\end{proof}      
\fi

\begin{remark} \label{rem:CFD}
The choice of the FVM and how we discretize the convection and diffusion terms appearing in~\eqref{eq:pdeOnMode}-\eqref{eq:pdeOffMode} is important for $A(t)$ to satisfy the conditions in Lemma~\ref{lem:rateMat}. This issue is well known in the CFD literature, and also recognized in the related work~\cite{BenenatiTractableCDC:2019}. If a finite difference method had been used with central differences for both diffusion and convection terms, the resulting $A(t)$ would require restrictive conditions on both $\sigma^2$ and $\Delta\lambda$ to satisfy the properties in Lemma~\ref{lem:rateMat}~\cite{VersteegIntroductionBook:2007}. 
\end{remark}

\subsection{Temporal discretization}
To temporally integrate the dynamics~\eqref{eq:dynContMC} we use a first order Euler approximation with time step $\Delta t>0$. Making the identifications $\nu_{k}\triangleq\nu(t_k)$ and $A_k \triangleq A(t_k)$ we have
\begin{align} \label{eq:discDynEsem}
	\nu_{k+1} &= \nu_kP_k, \quad \text{with} \quad P_k = I + \Delta tA_k.
\end{align}
In the continuous time setting elements of the vector $\nu(t)$ were referred to as, for example, $\nu_{\textOn}(\lambda^i,t)$. The counterpart to this, in the discrete time setting, is referring to elements of $\nu_k$ as, for example, $\nu_{\textOn}[\lambda^i,k]$. We further have the following.

\begin{lemma}\label{lem:CFL}
The matrix $P_k$ is a Markov transition probability matrix if
	\begin{align} \nonumber
	\forall \ i, \ \text{and} \ \forall \ k, \quad 0 < \Delta t \leq \left|[A_k]_{i,i}\right|^{-1}.
	\end{align}
        where $[A_k]_{i,i}$ is the $i^{th}$ diagonal element of the matrix $A_k$.
\end{lemma}
\begin{proof}
	From Lemma~\ref{lem:rateMat} we have that $P_k\mathbb{1} = I\mathbb{1} + \Delta tA_k\mathbb{1} = \mathbb{1}$ since $A_k\mathbb{1} = 0$. Also from Lemma~\ref{lem:rateMat}, every element of $A_k$ is non-negative, save for the diagonal elements. Under the hypothesis on $A_k$, then every diagonal element of $I + \Delta tA_k$ will be in $[0,1]$. 
\end{proof}

\begin{remark}\label{rem:N-vs-dt}
	The bound on the time step $\Delta t$ given in Lemma~\ref{lem:CFL} is $O(\Delta \lambda)$, which follows from the PDE discretization; \ifshowArxivAlt see~\cite[Appendix B]{CoffmanUnifiedArxiv:2021} \fi \ifshowArxiv see Appendix B \fi. Since $\Delta \lambda = \frac{\lambda^{\text{high}}-\lambda^{\textLB}}{N}$, as the temperature resolution $\Delta \lambda$ becomes finer the time resolution $\Delta t$ must also become finer at the same rate. See also Remark~\ref{rem:conditionsForCondFact} for a related comment.
\end{remark}

\section{Discrete space model of a TCL: structure and grid friendly policies}
\label{sec:indStruc}
Recall that the dynamics~\eqref{eq:discDynEsem} derived in the previous section was for the thermostat policy. We now delve into the structure of these dynamics so to introduce a BA control input. We first formalize a discrete state space for the dynamics~\eqref{eq:discDynEsem}. We will then show that the transition matrix in~\eqref{eq:discDynEsem} can be written as $P_k = \Phi G_k$ where $\Phi$ depends on the thermostat policy and $G_k$ on the TCL temperature dynamics and weather. The isolation of the policy then indicates how a BA could introduce grid friendly policies in place of the thermostat policy $\Phi$.

\subsection{Discrete state space}
When the conditions of Lemma~\ref{lem:CFL} are met $P_k$ is a transition matrix and hence each $\nu_k$ is a marginal pmf if $\nu_0$ is a pmf. The structure of this marginal is given from~\eqref{eq:probOnState} for the on state (a similar interpretation holds for the off state) as,
\begin{align} 
\nu_{\textOn}[\lambda^i,k] 
&=\Prob\left(\theta(t_k) \in \text{CV}(i), \ m(t_k) = \text{on} \right),
\end{align} 
where $\theta(t_k)$ is the temperature. Now denote, $\theta_k \triangleq \theta(t_k)$, $m_k \triangleq m(t_k)$, and
\begin{align} \label{eq:binnedState}
I_k \triangleq \sum_{i=1}^{N}i\mathbf{I}_{ \text{CV}(i)}(\theta_k,m_k).
\end{align}
The quantity $I_k$ indicates which CV the TCLs temperature resides in at time $k$. It also is a function of $m_k$ since the CV index for the on mode is different from the index for the off mode. We then define the following discrete state space:
\begin{align}
	\stateSetZ \triangleq \{m \in\{\textOn, \textOff\}, \ I \in \{1,\dots,N\} \},
\end{align}
with cardinality $\left|\stateSetZ\right| = 2N$.
Using the newly defined quantity $I_k$ we rewrite the marginals $\nu_{\textOn}[\lambda^i,k]$ and $\nu_{\textOff}[\lambda^i,k]$ as functions on $\stateSetZ$,
\begin{align} \label{eq:margNu}
\nu_{\textOn}[\lambda^i,k] &= \Prob\left(I_k = i, \ m_k = \text{on} \right), \quad \text{and} \\
\nu_{\textOff}[\lambda^i,k] &= \Prob\left(I_k = i, \ m_k = \text{off} \right).
\end{align}
From the above, the matrix $P_k$ (with the conditions of Lemma~\ref{lem:CFL} satisfied) is the transition matrix for the joint process $(I_k,m_k)$ on the state space $\stateSetZ$. The dynamic equation $\nu_{k+1} = \nu_kP_k$ is then a probabilistic model for a TCL with state space $\stateSetZ$ and operating under the thermostat policy. 

\subsection{Conditional independence in $P_k$} \label{sec:condIndPk}
 In the following, we refer to the values of $I_k$ with $i$ and $j$ and the values of $m_k$ with $u$ and $v$. We introduce the following notation to refer to the elements of the transition matrix $P_k$: 
\begin{align} 
  &   
P_{k}((i,u),(j,v)) \triangleq \\ \nonumber
 & \Prob\Big(I_{k+1} = j, \ m_{k+1} = v \ \Big\vert \ I_k = i, \  m_k = u, \ \theta^a_k = w_k\Big).
\end{align}
Recall, the matrix $P_k$ is derived for the thermostat policy. We will now show that the matrix $P_k$ can be written as the product of two matrices. One depends only on the thermostat policy (control) and the other depends only on weather and TCL temperature dynamics. That is, we show that each entry of $P_k$ factors as
\begin{align} \label{eq:condIndFact}
	P_{k}((i,u),(j,v)) = \phi^\text{TS}_u(v \ \vert \ i)P^u_{k}(i,j)
\end{align}
where, for each given values of $\theta^a_k$, $P^u_{k}(i,j)$ is a \emph{controlled transition matrix} on \stateSetZ:
\begin{align}\label{eq:Pu}	 
	P^u_{k}(i,j)\triangleq\Prob\left(I_{k+1} = j \ \vert \ I_k = i, \ m_k = u, \ \theta^a_k = w_k  \right)
\end{align}
and $\phi^\text{TS}_u(v \ \vert \ i)$ is an instance of a \emph{randomized policy} $\phi_u(v \ \vert \ i)$  on \stateSetZ:
\begin{align}\label{eq:detPol}	
	&\phi_u(v \ \vert \ i) \triangleq\Prob\left(m_{k+1} = v \ \vert \ I_{k} = i, \ m_{k} = u \right).
\end{align}
We show the factorization~\eqref{eq:condIndFact} through construction next.


\subsubsection{Constructing the factorization}
The quantity $\phi^\text{TS}_u(v \ \vert \ i)$ in~\eqref{eq:detPol} is the thermostat policy on $\stateSetZ$, which is formally defined as follows.
\begin{defn}\label{def:detPol}
	The thermostat policy on $\stateSetZ$ is specified by the two vectors, $\phi^\text{TS}_{\textOff},\phi^\text{TS}_{\textOn} \in \mathbb{R}^N$, where
$\phi^\text{TS}_{\textOff} \triangleq \phi^\text{TS}_\textOff(\textOn \ \vert \ \cdot) = \mathbf{e}_N$, 
$	\phi^\text{TS}_{\textOn} \triangleq \phi^\text{TS}_\textOn(\textOff \ \vert \ \cdot) = \mathbf{e}_1$, and $\phi^\text{TS}_\textOff(\textOff \ \vert \ \cdot) \triangleq 1 - \phi^\text{TS}_\textOff$, $\phi^\text{TS}_\textOn(\textOn \ \vert \ \cdot) \triangleq 1 - \phi^\text{TS}_\textOn$.
\end{defn}
The quantity $P^u_{k}(i,j)$ in~\eqref{eq:Pu} represents the policy-free (open loop) evolution of the TCL on $\stateSetZ$. That is, it describes how the TCLs temperature evolves under a fixed mode. We define matrices with entries $P^u_{k}(i,j)$ next.
\begin{defn} \label{def:contFree}
	Let $P_k^\textOff,P_k^\textOn\in\mathbb{R}^{N\times N}$ have $(i,j)$ entries given by,
	\begin{align} \nonumber
	P_k^\textOff(i,j) &= P_{w_k}((i,\textOff),(j,\textOff)), \quad i \neq N \ \text{and} \ j \neq N, \\ \nonumber
	P_k^\textOn(i,j) &= P_{w_k}((i,\textOn),(j,\textOn)), \quad i \neq 1 \ \text{and} \ j \neq 1,
	\end{align}
	with $P_k^\textOff(N,N) = 1$ and $P_k^\textOn(1,1) = 1$.
\end{defn}
The quantities defined in Definition~\ref{def:detPol} and~\ref{def:contFree} correspond to entries of $P_k$. To construct the promised factorization, from these definitions, the idea is to construct its four sub-matrices that correspond to all possible combinations of $u,v \in \{\textOn,\textOff\}$ (see Figure~\ref{fig:sparPattern}). For example, the $\textOff-\textOff$ quadrant of $P_k$ is given by the matrix product
\begin{align} \nonumber
	\big(I - \text{diag}(\phi^\text{TS}_\textOff)\big)P^\textOff_k.
\end{align}
However, since the temperature associated with the $i^{th}$ CV for the on mode is not the same temperature associated with the $i^{th}$ CV for the off state (see Figure~\ref{fig:cvLayout}) it is \emph{not} true that the $\textOff-\textOn$ quadrant of $P_k$ is given as $\text{diag}(\phi^\text{TS}_\textOff)P^\textOff_k$. The entries of the matrix $P^\textOff_k$ need to be re-arranged so to correctly account for the difference in CV index between the on/off mode. We define such correctly re-arranged matrices next.

\begin{defn} \label{def:Smat}
	Let $I^\textOff = \{m,\dots,N\}$, $I^\textOn = \{1,\dots,q \}$, $m^- = m-1$, and $S_k^\textOff,S_k^\textOn \in \mathbb{R}^{N\times N}$ with $(i,j)$ entries
	\begin{align}
		S_k^\textOff(i,j-m^-) &= \begin{cases}
		P_k^\textOff(i,j) & i,j \in I^\textOff \\
		0 & \text{otherwise}.
		\end{cases} \\
		S_k^\textOn(i,j+m^-) &= \begin{cases}
		P_k^\textOn(i,j) & i,j \in I^\textOn \\
		0 & \text{otherwise}.
		\end{cases}
	\end{align}
\end{defn}
The above definition is based on the construction that $N =q+m$.
The quantities in Definition~\ref{def:Smat} let us construct, e.g., the $\textOff-\textOn$ quadrant of $P_k$ as $\text{diag}(\phi_{\textOff}^\text{TS})S_k^\textOff$.

The next result shows that $P_k = \Phi^\textTS G_k$ under certain conditions and for appropriate choices of the matrices $\Phi^\textTS$ and $G_k$.

\begin{lemma} \label{lem:condFact}
  	Let the time discretization period $\Delta t$ and the parameter $\alpha$ that appears as a design choice in discretizing the PDEs to ODEs be chosen to satisfy $\alpha = (\Delta t)^{-1}$. Let $\Phi^\text{TS}_{\textOff} =\text{diag}(\Phi^\text{TS}_{\textOff})$ and $\Phi^\text{TS}_{\textOn}=\text{diag}(\Phi^\text{TS}_{\textOn})$, and
	\begin{align} \label{eq:polMatInLem}
		\Phi^\textTS & \triangleq \begin{bmatrix}
			I - \Phi^\text{TS}_{\textOff} & \Phi^\text{TS}_{\textOff} & \mathbf{0} & \mathbf{0} \\
			\mathbf{0} & \mathbf{0} & \Phi^\text{TS}_{\textOn} & I - \Phi^\text{TS}_{\textOn}
		\end{bmatrix} \quad \text{and} \\
		G_k & \triangleq \begin{bmatrix}
		\mathbf{0} & S_k^{\textOff} & \mathbf{0} & P_k^{\textOn} \\
		P_k^{\textOff} & \mathbf{0} & S_k^{\textOn} & \mathbf{0}
		\end{bmatrix}^T,
	\end{align}
	then
        \begin{align}
          \label{eq:condIndFactMat}
          P_k = \Phi^\textTS G_k.
        \end{align}
      \end{lemma}
\ifshowArxivAlt
\begin{proof}
	See Appendix in~\cite{CoffmanUnifiedArxiv:2021}.
\end{proof}	
\fi
\ifshowArxiv
\begin{proof}
	See Appendix~\ref{app:condFactProof}.
\end{proof}	
\fi

\begin{remark}\label{rem:conditionsForCondFact}
The condition $\alpha = 1/\Delta t$ can be satisfied as long as time and temperature discretization intervals are chosen to satisfy $\Delta t < (\Delta \lambda)^2/\sigma^2$.  To understand how, recall that in the discretizing the PDE to the coupled ODEs, a design parameter $\gamma>0$ appears: some rate of density is transferred out of the control volume $\lambda^N_\textOff$ and into the CV $\lambda^q_\textOn$ (as depicted in Figure~\ref{fig:cvLayout}) due to thermostatic control. The rate of the density transfer is then given as $-\gamma\nu_{\textOff}(\lambda^{N},t)$, where $\gamma>0$ is a modeling choice and a constant of appropriate units that describes the discharge rate. We then define $\alpha \triangleq D + \gamma$ where $D = \frac{\sigma^2}{(\Delta \lambda)^2}$. Recall that $\sigma^2$ is the variance in the Fokker-Planck equation~\eqref{eq:pdeOnMode}-\eqref{eq:pdeOffMode} and $\Delta \lambda$ is the temperature discretization interval. Thus, as long as $1/\Delta t > D$, a positive $\gamma$ can be chosen while meeting the condition $\alpha = 1/\Delta t$. The condition $1/\Delta t > D$ is equivalent to $\Delta t < (\Delta \lambda)^2/\sigma^2$.
\end{remark}

        \begin{remark}\label{rem:cond-indep}
The conditional independence factorization~\eqref{eq:condIndFactMat} has been a useful assumption in the design of algorithms in~\cite{busmey:CDC:2016}. In the present it is a byproduct of our spatial and temporal discretization of the PDEs~\eqref{eq:pdeOnMode}-\eqref{eq:pdeOffMode}. There are other works~\cite{BenenatiTractableCDC:2019,BashashModelingTCST:2013,PaccagnanRangeCDC:2015} that develop Markov models for TCLs through discretization of PDEs. However, to our knowledge, our work is the first to uncover this factorization.           
        \end{remark}
Lemma~\ref{lem:condFact} informs us how to define the dynamics of the marginals~\eqref{eq:margNu} under a different policy than the thermostat policy, which is described next.

\subsection{BA control command $=$ policy} \label{sec:contAggMod}
In light of the previous section, an arbitrary randomized policy can replace the thermostat policy to control the state process on $\stateSetZ$. That is equivalent to replacing $\Phi^\textTS$ in~\eqref{eq:condIndFactMat} with a new matrix $\Phi$ that corresponds to a policy designed for grid support. From the viewpoint of the BA this  randomized policy \emph{is} the control input that it must design and broadcast to a TCL. The TCL now implements this policy to make on/off decisions instead of using the thermostat policy. As we shall soon see, if the BA appropriately designs and sends the randomized policy to multiple TCLs it can achieve coordination of the TCLs for grid support.

To distinguish from thermostat policy $\phi^\text{TS}_{\textOff}$ and $\phi^\text{TS}_{\textOn}$ in the prior section that only maintains temperature, we denote the newly introduced policies for providing grid support with the superscript `$\textPolName$'. We require the policies, $\phi^{\textPolName}_{\textOn}$ and $\phi^{\textPolName}_{\textOff}$, to have the following structure
\begin{align} \label{eq:randPolOff2On}
&\phi^{\textPolName}_{\textOff}(\text{on} \ \vert \ j) = \begin{cases}
	\kappa^{\textOn}_j, & (m+1) \leq j \leq (N-1). \\
	1, & j = N. \\
	0, & \text{o.w.}
	\end{cases} \\ \label{eq:randPolOn2Off}
	&\phi^{\textPolName}_{\textOn}(\text{off} \ \vert \ j) = \begin{cases}
	\kappa^{\textOff}_j, & 2 \leq j \leq (q-1). \\
	1, & j = 1. \\
	0, & \text{o.w.}
	\end{cases}
\end{align}
with $\phi^{\textPolName}_{\textOff}(\text{off} \ \vert \ \cdot) = 1 - 	\phi^{\textPolName}_{\textOff}(\text{on} \ \vert \ \cdot)$ and $\phi^{\textPolName}_{\textOn}(\text{on} \ \vert \ \cdot) = 1 - 	\phi^{\textPolName}_{\textOn}(\text{off} \ \vert \ \cdot)$ and $\kappa^{\textOn}_j,\kappa^{\textOff}_j \in [0,1]$ for all $j$. The policies could also be time varying, for example: $\kappa^{\textOff}_j[k]$ and $\kappa^{\textOn}_j[k]$. The dependence of the policies on time is denoted as $\phi^{\textPolName}_{\textOff}[k]$ and $\phi^{\textPolName}_{\textOn}[k]$. Designing the grid support control policies is then equivalent to choosing the values of $\kappa^\textOn_j[k]$ and $\kappa^\textOff_j[k]$ for all $j$ and $k$.

We have required $\phi^{\textPolName}_{\textOff}(\text{on} \ \vert \ j) = 0$ for $1\leq j \leq m$ since the temperatures corresponding to these indices are below the permitted deadband temperature, $\lambda^{\text{min}}$. Hence, turning on at these temperature does not make physical sense. The arguments for the zero elements in $\phi^{\textPolName}_{\textOn}$ are symmetric.

\begin{remark} \label{rem:impRandPol}
From the individual TCL's perspective, implementing grid support randomized policies of the form~\eqref{eq:randPolOff2On}-\eqref{eq:randPolOn2Off} is straightforward: (i) the TCL measures its current temperature and on/off status, (ii) the TCL ``bins'' this temperature value according to~\eqref{eq:binnedState} and (iii) the TCL flips a coin to decide its next on/off state according to the probabilities given in~\eqref{eq:randPolOff2On}-\eqref{eq:randPolOn2Off}. Note that the thermostat policy is a special case of the grid support control policy, and both policies enforce the temperature constraint. 
\end{remark}

\section{Proposed framework} \label{sec:propFrame}
We are now in a position to present our unified framework for coordination of TCLs. We first expand the state of the model~\eqref{eq:discDynEsem} so to incorporate cycling, following~\cite{liushi:2016,TotuDemandCST:2017}. We then shift the viewpoint from a single TCL to that of a collection of TCLs (recall Remark~\ref{rem:LLN}) to develop our control oriented aggregate model. Using this model we develop a method for designing both reference and policy through convex optimization.  
\subsection{Individual TCL model with cycling} \label{sec:lockedDyn}
 We now augment the model for a TCL's temperature evolution with cycling dynamics. Recall the cycling constraint: as soon as a TCL switches its mode, the TCL becomes stuck in that mode for $\tau$ time instances. This constraint can be represented as the evolution of a state, specifically, a counter variable. First defining the binary variable $s_k$ as $s_k=1$ if the TCL is stuck in the current mode at time $k$ and $0$ if it is not stuck. The counter variable is defined as follows
\begin{align}
\Locked_{k+1} \triangleq \begin{cases}
\Locked_{k} + 1, & s_k = 1. \\
0, & s_k = 0.  
\end{cases}
\end{align}
This variable denotes the time spent in the ``stuck'' mode ($s_k = 1$). A TCL has flexibility to help the grid only when $\Locked_k = 0$, which means it is not stuck in either the on or off mode. If $\Locked_k >0$, it is stuck in either the on or off mode, and switching the mode to help the grid will violate the cycling constraint. 

Recall, the discrete state space $\stateSetZ$ for a TCL included binned temperature and on/off mode. The space $\stateSetZ$, the policies $\phi^\textPolName_\textOn$ and $\phi^\textPolName_\textOff$, the marginal pmf $\nu_k$, and the transition matrix $P_k$ (and consequently its factors $\Phi$ and $G_k$) now all need to be expanded to be defined over a state space consisting of $(I_k,m_k,\Locked_k)$. This expansion is described next.

We denote this newly expanded state space as the set of values: $\stateSet \triangleq $
\begin{align}
\Big\{m\in\{\textOn,\textOff\}, \ I \in \{1,\dots,N\}, \ \Locked \in \{0,\dots,\tau\}  \Big\},
\end{align}
with cardinality $|\stateSet| = 2N(\tau+1)$. The policies on the expanded state space are:
\begin{align} \label{eq:expPolStruct}
 \phi^{\text{E}}_{\textOff} &= \mathbf{I}_{\{ 0\}}(\Locked)\phi^{\textPolName}_{\textOff} + (1-\mathbf{I}_{\{ 0\}}(\Locked))\phi^{\text{TS}}_{\textOff}, \quad \text{and} \\ \nonumber \phi^{\text{E}}_{\textOn} &= \mathbf{I}_{\{ 0\}}(\Locked)\phi^{\textPolName}_{\textOn} + (1-\mathbf{I}_{\{ 0\}}(\Locked))\phi^{\text{TS}}_{\textOn}.
\end{align}
To ensure that expanded policy~\eqref{eq:expPolStruct} will enforce the cycling constraint, we impose the following restriction at the design stage: a TCL with $\Locked_k > 0$ will only implement the thermostat policy, and a TCL with $\Locked_k = 0$ will make on/off decisions based on the grid support policy. The construction in this way ensures a TCL will not violate its cycling and temperature constraints under the conditions in Assumption \textbf{A.2} and \textbf{A.3}.

Each entry of the expanded policy is denoted as $\phi^\text{E}_\textOff(u \ \vert \ j, \ l)$ and $\phi^\text{E}_\textOn(u \ \vert \ j, \ l)$. The expanded marginals are $\nu_{\textOff}[\lambda^j,l,k]$ and $\nu_{\textOn}[\lambda^j,l,k]$, and $\nu_{\textOff,l}$ (resp., $\nu_{\textOn,l}$) is shorthand for $\nu_{\textOff}[\cdot,l,k]$ (resp., $\nu_{\textOn}[\cdot,l,k]$). In vectorized form, the expanded marginal is $\nu^\text{E} = [\nu_{\textOff}^\text{E},\nu_{\textOn}^\text{E}]$ where $\nu_{\textOff}^\text{E} = [\nu_{\textOff,0}, \dots, \nu_{\textOff,\tau}]$ and $\nu_{\textOn}^\text{E} = [\nu_{\textOn,0}, \dots, \nu_{\textOn,\tau}]$. Define
	\begin{align} \label{eq:polMatE}
G^\text{E}_k &\triangleq \begin{bmatrix}
\mathbf{0} & D_\tau \otimes S_k^{\textOn} & \mathbf{0} & C_\tau \otimes P_k^{\textOff} \\
C_\tau \otimes P_k^{\textOn} & \mathbf{0} & D_\tau \otimes S_k^{\textOff} & \mathbf{0}
\end{bmatrix}^T,
\end{align}
where $D_\tau \triangleq \mathbb{1}^T \otimes \mathbf{e}_2\in\mathbb{R}^{\tau+1 \times \tau + 1}$ and  
\begin{align}
	C_\tau \triangleq \begin{bmatrix}
	1 & 0 & \mathbf{0}_{\tau-1}^T \\
	\mathbf{0}_{\tau-1} & \mathbf{0}_{\tau-1} & I_{\tau-1} \\
	1 &  0 & \mathbf{0}_{\tau-1}^T
	\end{bmatrix} \in \mathbb{R}^{(\tau+1) \times (\tau + 1)}.
\end{align}
We define the matrix $\Phi_k^\text{E}$ as having the same structure as~\eqref{eq:polMatInLem}, but with the expanded policies $\phi^{\text{E}}_{\textOff}$ and $\phi^{\text{E}}_{\textOn}$, i.e.,
\begin{align} \label{eq:fullBAcontPol}
\Phi^{\text{E}}_k \triangleq \begin{bmatrix}
I - \Phi^{\text{E}}_{\textOff}[k] & \Phi^{\text{E}}_{\textOff}[k] & \mathbf{0} & \mathbf{0} \\
\mathbf{0} & \mathbf{0} & \Phi^{\text{E}}_{\textOn}[k] & I - \Phi^{\text{E}}_{\textOn}[k]
\end{bmatrix},
\end{align} 
where $\Phi^{\text{E}}_{\textOff}[k] \triangleq \text{diag}(\phi^{\text{E}}_{\textOff}[k])$ and $\Phi^{\text{E}}_{\textOn}[k] \triangleq \text{diag}(\phi^{\text{E}}_{\textOn}[k])$. The model of a TCL with cycling dynamics and grid support policy becomes 
\begin{align} \label{eq:indProbModel}
	\nu^{\text{E}}_{k+1} = \nu^{\text{E}}_{k}\Phi_k^\text{E}G_k^\text{E}. 
\end{align}
The structure of the transition matrix $\Phi_k^\text{E}G_k^\text{E}$ is shown in Figure~\ref{fig:sparPatCycle}. For comparison, the transition matrix with policy $\phi^\textPolName$ and \emph{without} the cycle counter variable would simply be the four red shaded blocks appearing in their respective quadrant. In the expanded system, an on to off mode switch forces probability mass from the red shaded region ($l=0$ and $m=\textOn$) to the green shaded region ($l=1$ and $m=\textOff$). Mass must then transition through the chain of $\tau$ green blocks until it reaches the red block again, so to respect the cycling constraint.

\begin{figure}
	\centering
	\includegraphics[width=1\columnwidth]{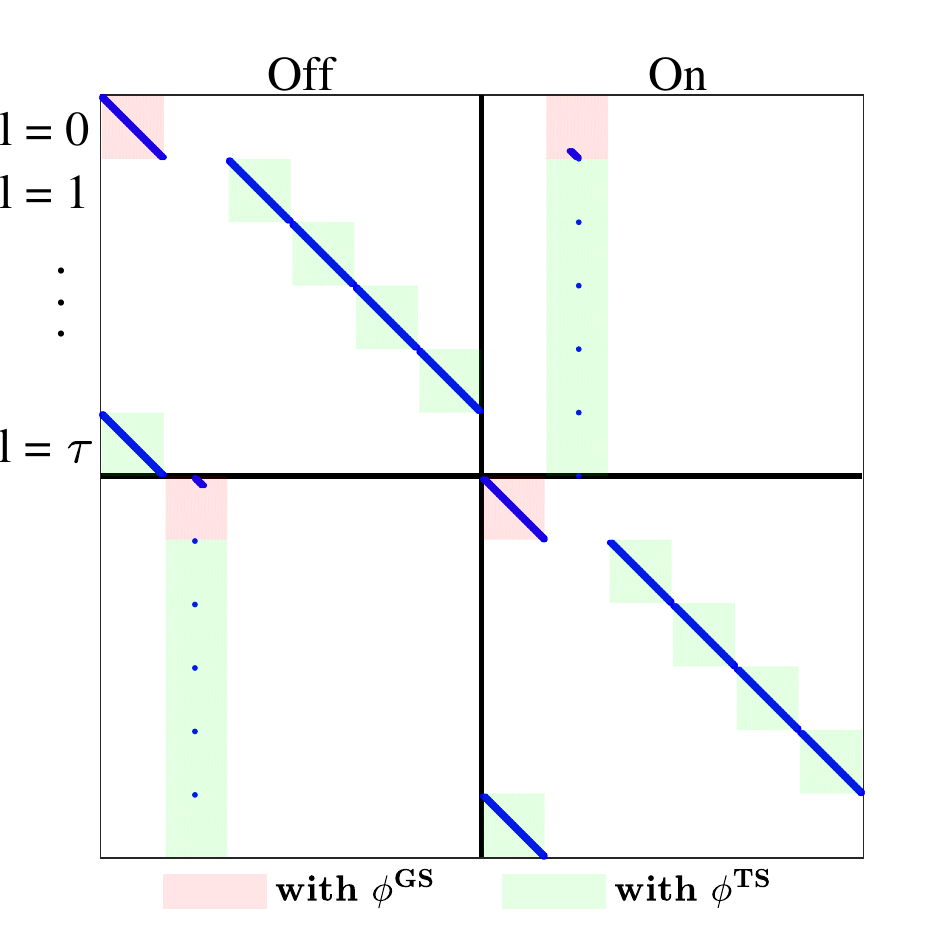}
	\caption{The sparsity pattern of the expanded transition matrix (the dots represent non-zero entries in the matrix) with $\tau = 5$. Each shaded block is over the entire range of temperature values.}
	\label{fig:sparPatCycle}
\end{figure}

\subsection{Aggregate model of a collection of TCLs} \label{sec:aggModel}
We now transition from the viewpoint of a single TCL to that of a collection of $\numTCLs$ TCLs: $\ell = 1,\dots,\numTCLs$. For example, $m_k^\ell$ and $I_k^\ell$ are the mode and binned temperature of the $\ell^{th}$ TCL at time $k$. Recall Remark~\ref{rem:LLN}, the model~\eqref{eq:indProbModel} also describes an entire collection of TCLs. For a single TCL, we view the state $\nu_k^\text{E}$ as a marginal but for a collection of TCLs we expect the marginal pmf $\nu_k^\text{E}$ to approximate the histogram
\begin{align} \label{eq:empHist}
 h_k[u,i,l] &\triangleq \frac{1}{\numTCLs}\sum_{\ell=1}^{\numTCLs}\Big( \mathbf{I}_{\{i\}}(I^{\ell}_k)\mathbf{I}_{\{u\}}(m^{\ell}_k)\mathbf{I}_{\{l\}}(\Locked^{\ell}_k)\Big), 
\end{align}
for each state $(u,i,l) \in \stateSet$ as $\numTCLs\rightarrow\infty$. In the same regard, we define 
\begin{align} \label{eq:outAggModel}
	\gamma_k^\text{E} \triangleq \nu_k^\text{E}C^\text{E}, \quad \text{where} \quad C^\text{E} \triangleq [\mathbf{0}^T,P_\textAgg\mathbb{1}^T]^T,
\end{align}
where $P_\textAgg$ is the maximum possible power of the collection, defined in~\eqref{eq:powerAgg}. We expect $\gamma_k^\text{E}$ to approximate the total power consumption $y_k$ of the collection of $\numTCLs$ TCLs:
\begin{align} \label{eq:totPow}
	y_k &\triangleq P\sum_{\ell=1}^{\numTCLs}m_k^\ell.
\end{align}
which is the discrete-time equivalent of $y(t)$ defined in~\eqref{eq:y-def}. That is, we expect $\gamma_k^\text{E} \approx y_k$ for large $\numTCLs$, based on a law of large numbers argument~\cite{ChenStateTAC:2017}. The \emph{control oriented aggregate model of a TCL collection} is the dynamics~\eqref{eq:indProbModel} together with the output~\eqref{eq:outAggModel}:
\begin{align} \label{eq:expAggModel}
	\nu^{\text{E}}_{k+1} = \nu^{\text{E}}_{k}\Phi_k^\text{E}G_k^\text{E}. \qquad \text{and} \qquad \gamma_k^\text{E} = \nu_k^\text{E}C^\text{E}. 
\end{align}

\ifshowArxivAlt
Effectiveness of~\eqref{eq:expAggModel} in modeling a population of TCLs can be seen in our prior work~\cite{CoffmanControlACC:2021}.
\fi

\ifshowArxiv
\subsubsection{Evaluating the aggregate model} \label{sec:evalAggModel}
Before proceeding to policy design with our developed model~\eqref{eq:expAggModel}, we first show that it is effective in modeling a population of TCLs. We do this by comparing the state of the model to~\eqref{eq:empHist} and~\eqref{eq:totPow} obtained from a simulation of \numTCLs = 50,000 air conditioning TCLs. 

The comparison results are shown in Figure~\ref{fig:histSimCompare} and Figure~\ref{fig:BAcontPol}. The mode state of each TCL evolves according to a control policy, where the $\phi^\textPolName_\textOff$ and $\phi^\textPolName_\textOn$ portion are shown in Figure~\ref{fig:BAcontPol} (bottom). The policy is arbitrary, designed merely to be an example of a non-thermostat policy. This policy satisfies the structure in~\eqref{eq:randPolOff2On} and~\eqref{eq:randPolOn2Off} so that both temperature and cycling constraints are satisfied at each TCL. The temperature evolution evolves according to~\eqref{eq:stoModelTCL}. We see the state $\nu^\text{E}_k$ matches the histogram $h_k$ of the collection for the devices that are not stuck (Figure~\ref{fig:histSimCompare} (top)) and for the devices that are stuck (Figure~\ref{fig:histSimCompare} (bottom)). Additionally, the output of the aggregate model, $\gamma_k^{\text{E}}$, matches it's empirical counterpart $y_k$ (shown in Figure~\ref{fig:BAcontPol} (top)).

\begin{figure}
	\centering
	\begin{minipage}{0.5\textwidth}
		\includegraphics[width = 1\columnwidth]{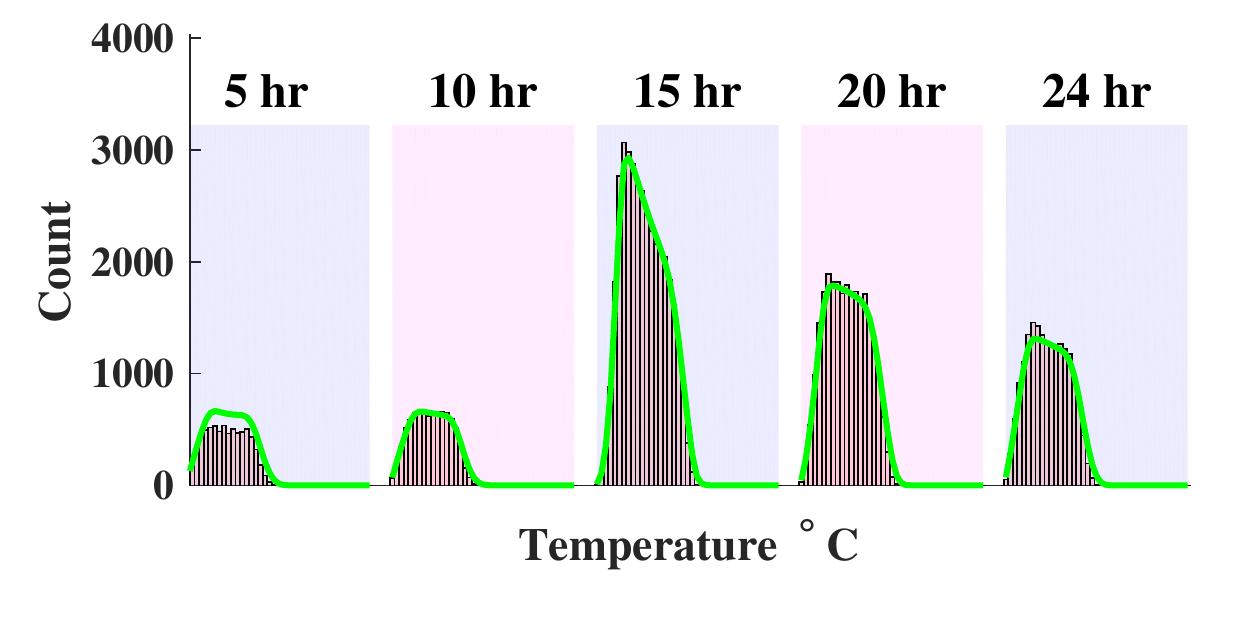}
	\end{minipage}
	\begin{minipage}{0.5\textwidth}
		\includegraphics[width = 1\columnwidth]{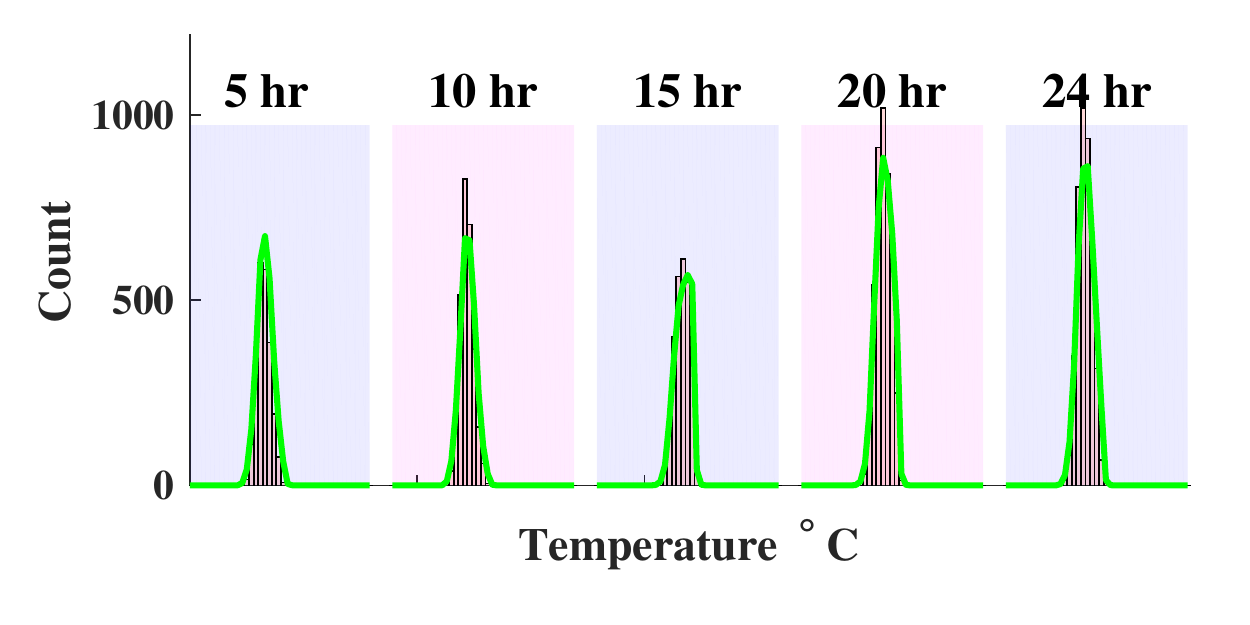}
	\end{minipage}
	\caption{(Top): Histogram of the collection for the devices that are on and not stuck. (Bottom): Histogram of the collection for the devices that are on and are stuck.}
	\label{fig:histSimCompare}
\end{figure}

\begin{figure}
	\centering
	\begin{minipage}{0.5\textwidth}
		\includegraphics[width = 1\columnwidth]{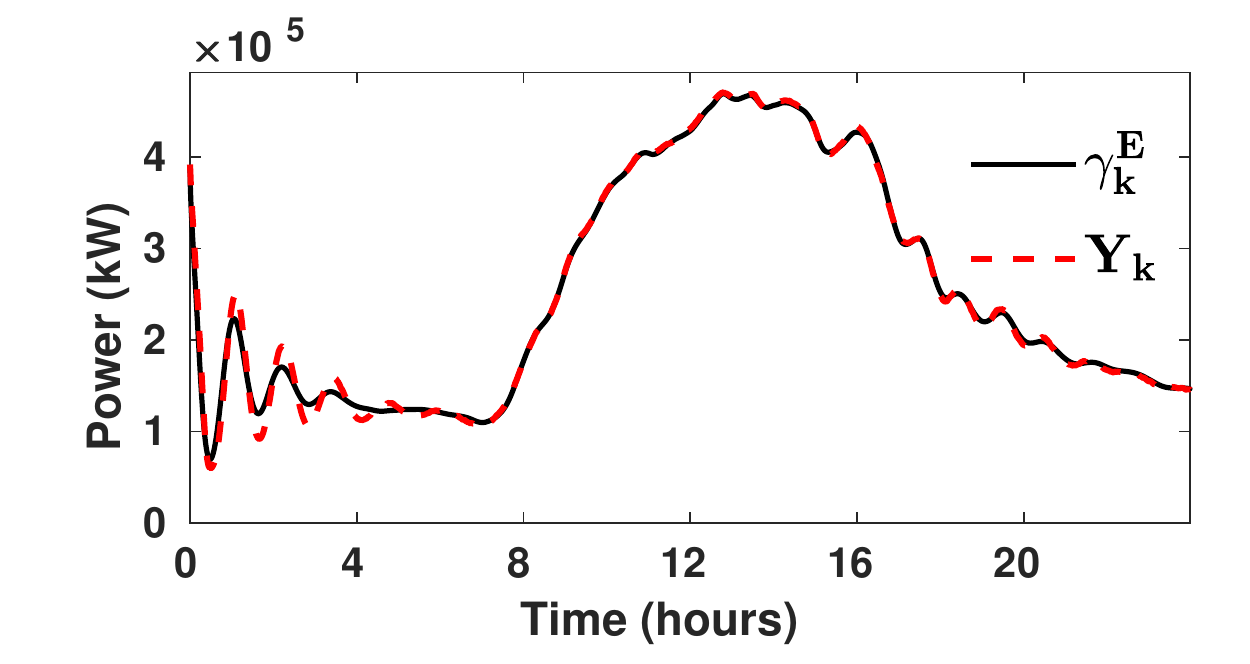}
	\end{minipage}
	\begin{minipage}{0.5\textwidth}
		\includegraphics[width = 1\columnwidth]{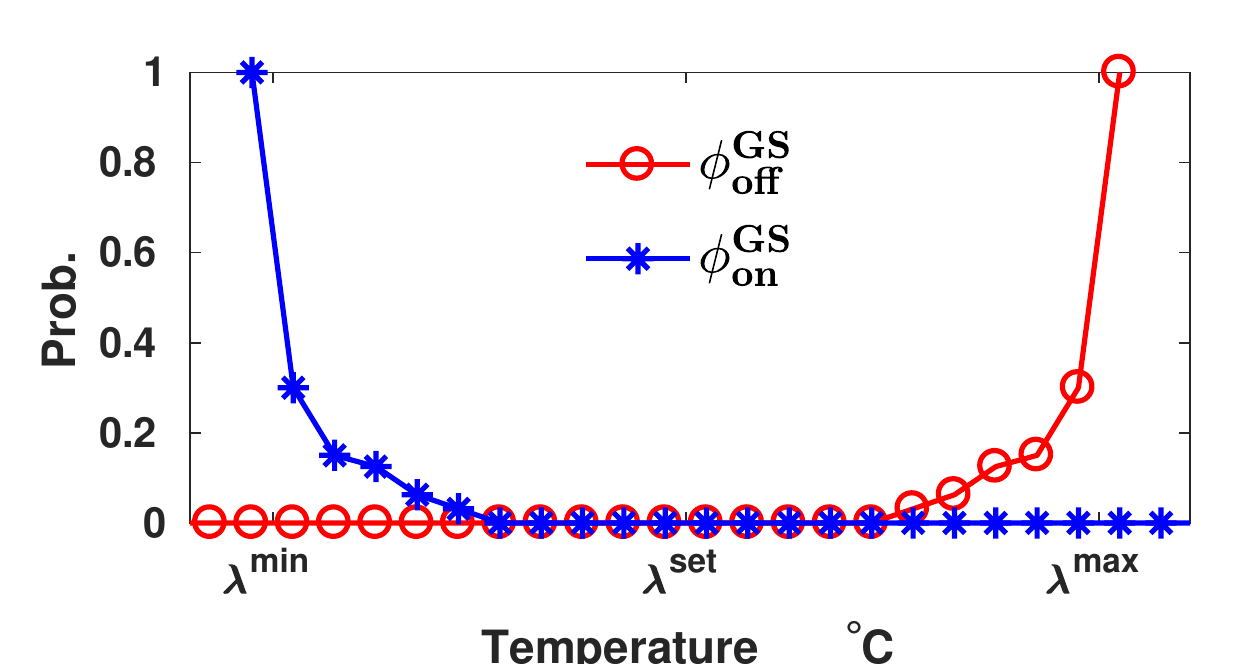}
	\end{minipage}	
	\caption{(Top): Comparison of the output of the expanded aggregate model $\gamma_k^\text{E}$ and the ensembles power consumption $y_k$. (Bottom): The policies $\phi^\textPolName_\textOff$ and $\phi^\textPolName_\textOn$ used for the numerical experiment in Section~\ref{sec:aggModel}.}
	\label{fig:BAcontPol}
\end{figure}
\fi

\subsection{Grid support Policy design} \label{sec:polDesign}

The goal of coordinating TCLs is to help the BA balance supply and demand of electricity in the grid. We denote $r^\textBA_k$ as the desired demand from all flexible loads and batteries that will reduce the imbalance to 0. It is unreasonable to expect any collection of TCLs to meet the entire desired demand $r^\textBA_k$ while maintaining their QoS. Only a portion of $r^\textBA_k$ can be supplied by TCLs, and we denote this portion by $r_k$. Determining $r_k$ becomes an optimal control problem due to the time coupling produced by the TCL dynamics. We consider a planning horizon of $T_{\text{plan}}$. To simultaneously design grid support control policies $\phi^{\textPolName}_\textOff[k]$ and $\phi^{\textPolName}_\textOn[k]$ and determine a suitable reference signal $r_k$ over $T_{\text{plan}}$ the BA solves the following optimization problem,
\begin{align} \label{eq:nonConvOpt}
	\eta^* = \min_{\nu^\text{E}_k,\Phi^{\text{E}}_k} \ &\eta(\hat{\nu}) = \sum_{k\in\timeHorz}\Big(r^\textBA_k - \gamma^\text{E}_k\Big)^2  \\
	\text{s.t.} \ &\nu^{\text{E}}_{k+1} = \nu^{\text{E}}_{k}\Phi_k^\text{E}G_k^\text{E}, \quad \nu^\text{E}_{\timeHorz(0)} = \hat{\nu}, \\
				&\gamma^\text{E}_k = \nu_k^\text{E}C^\text{E}, \quad \nu^{\text{E}}_k \in [0,1], \quad \Phi^{\text{E}}_k \in \varPhi.
\end{align}
The solution at time $k$ is denoted $r_k \triangleq \gamma_k^{\text{E},*}$, $\phi^{\textPolName,*}_\textOff[k]$, and $\phi^{\textPolName,*}_\textOn[k]$.
We have $\timeHorz \triangleq \{\timeHorz(0), \dots, \timeHorz(0) + T_{\text{plan}}-1\}$ is the index set of times, $\timeHorz(0)$ denotes the initial time index, $\hat{\nu}$ is the initial condition, and $\nu_k^\text{E}\in[0,1]$ holds elementwise. The set $\varPhi$ collects all of the constraints on the policy. This includes the equality constraints set by the structural requirements in~\eqref{eq:randPolOff2On}-\eqref{eq:randPolOn2Off} and~\eqref{eq:expPolStruct} as well as the structural requirement in~\eqref{eq:fullBAcontPol}. These constraints require certain elements of the policy to be either zero or one. The policy should also be a valid conditional pmf and its elements in $[0,1]$. Hence, the set $\varPhi$ is the following convex set
\begin{align} \nonumber
\varPhi \triangleq \Big\{\Phi \in \mathbb{R}^{|\stateSet|\times 2|\stateSet|}_{[0,1]}\ \big\vert \ &\Phi \ \text{satisfies}~\eqref{eq:fullBAcontPol}, \ \mathbb{1} = \Phi\mathbb{1},  \\ \nonumber
&\phi^{\textPolName}_{\textOff} \ \text{satisfies}~\eqref{eq:randPolOff2On}, \\ \nonumber
&\phi^{\textPolName}_{\textOn} \ \text{satisfies}~\eqref{eq:randPolOn2Off}, \ \text{and} \\
&\phi^{\text{E}}_{\textOff} \ \text{and} \ \phi^{\text{E}}_{\textOn} \ \text{satisfy}~\eqref{eq:expPolStruct}\Big\}.
\end{align}
Where, e.g., $\mathbb{R}^{|\stateSet|\times|\stateSet|}_{[0,1]}$ is the set of $|\stateSet|\times|\stateSet|$ matrices with elements in $[0,1]$. 
\ifx 0
\begin{align} \nonumber
	\varPhi \triangleq \Big\{&\Phi \in \mathbb{R}^{|\stateSet|\times 4|\stateSet|}_{[0,1]} \ \Big\vert  \ \Phi \ \text{satisfies}~\eqref{eq:fullBAcontPol}, \ \mathbb{1} = \Phi\mathbb{1}, \\ \nonumber &\phi^\text{E}_\textOff(u \ \vert \ j, \ l) = \beta_\textOff(u,j,l), \quad \forall \ (u,j,l) \in \mathcal{W}_\textOff,\\ \nonumber
	&\phi^\text{E}_\textOn(u \ \vert \ j, \ l) = \beta_\textOn(u,j,l), \quad \ \forall \ (u,j,l) \in \mathcal{W}_\textOn. \Big\}.
\end{align}
Where, e.g., $\mathbb{R}^{|\stateSet|\times|\stateSet|}_{[0,1]}$ is the set of $|\stateSet|\times|\stateSet|$ matrices with elements in $[0,1]$. 
\fi

\paragraph{QoS + Solution of~\eqref{eq:nonConvOpt} }
	\medskip

	\begin{enumerate}
		\item
		
		The equality constraints in $\Phi$ are present to ensure the individual TCL's QoS constraints: the structure~\eqref{eq:expPolStruct} ensures the cycling constraint and the structure~\eqref{eq:randPolOff2On}-\eqref{eq:randPolOn2Off} ensures the temperature constraint. Recall that this structure guarantees QoS by requiring the policy to place zero probability on state transitions that would violate QoS.
		\item A solution to~\eqref{eq:nonConvOpt} yields, for $k\in\timeHorz$, two things: (i) the optimal randomized policies $\phi^{\textPolName,*}_\textOff[k]$ and $\phi^{\textPolName,*}_\textOn[k]$ and (ii) an optimal reference for the power demand of the TCL collection $r_k (= \gamma^{\text{E},*}_k)$. The reference is optimal in the following sense: among all power demand signals the collection can track without requiring any TCL to violate its local QoS constraints in so doing, it is the closest to the BA's desired demand $r^\textBA$ in 2-norm.  The reference is also the predicted power consumption of the TCLs whilst using the policies $\phi^{\textPolName,*}_\textOff[k]$ and $\phi^{\textPolName,*}_\textOn[k]$.
	\end{enumerate}
 
\begin{remark} \label{rem:tclCap}
	Since the reference $r_k (= \gamma^{\text{E},*}_k)$ from~\eqref{eq:nonConvOpt} is the best the TCLs can do to help the BA without any TCL having to violate its QoS, Problem~\eqref{eq:nonConvOpt} therefore also provides an answer to the ``aggregate flexibility'' question: how much can a collection of TCLs vary their demand while maintaining their local QoS constraints. This question has been investigated by many works~\cite{PaccagnanRangeCDC:2015,CoffmanCharacterizingTPS:2020,coffmanFlexibilityTCL_ArXiV:2020,hao_aggregate:2015}.
\end{remark}


\subsubsection{Convex control synthesis}
The problem~\eqref{eq:nonConvOpt} is non-convex due to the product  $\nu^{\text{E}}_{k}\Phi_k^\text{E}$ in the constraint.
A well known convexification remedy for~\eqref{eq:nonConvOpt} is to consider optimizing over the marginal and joint distribution instead of the marginal and the policy~\cite{ManneLinearINFORMS:1960,BenenatiTractableCDC:2019}. Using our identified structure from Section~\ref{sec:condIndPk} we construct the following joint distribution (written in matrix form):
\begin{align} \label{eq:changVar}
	J_k = \text{diag}(\nu^{\text{E}}_k)\Phi^{\text{E}}_k \in \mathbb{R}^{|\stateSet| \times 2|\stateSet|}.
\end{align}
By construction, we have that $\nu^{\text{E}}_{k+1} = \mathbb{1}^TJ_kG^\text{E}_k$ and $(\nu^{\text{E}}_k)^T = J_k\mathbb{1}$ since $\mathbb{1}^T\text{diag}(\nu^{\text{E}}_k) = \nu^{\text{E}}_k$  and $\mathbb{1} = \Phi^{\text{E}}_k\mathbb{1}$. It is straightforward to convert the constraint set $\Phi^{\text{E}}_k \in \varPhi$ to the new decision variables. For the equality constraints in $\varPhi$ if we have that $\phi^\text{E}_\textOff(u \ \vert \ j,l) = \kappa$, then in the decision variables $J_k$ and $\nu_k^\text{E}$ we will have a linear constraint of the form
\begin{align} \nonumber
	&\Prob\left(m_{k+1} = u,\ I_{k} = j,\ \Locked_k = l,  \ m_k = \text{off}\right) \\ \label{eq:polCons} 
	&=\kappa\nu_{\textOff}[\lambda^j,l,k], 
\end{align}
where the LHS of the above is some element in the matrix $J_k$. In addition to the above equality constraints, requiring both $J_k$ and $\nu_k^\text{E}$ to be within $[0,1]$ and the constraint $(\nu^\text{E}_k)^T = J_k\mathbb{1}$ will allow one to reconstruct a policy $\Phi_k^\text{E}\in\varPhi$ from $J_k$ and $\nu_k^\text{E}$ (described shortly in Lemma~\ref{lem:polAlg}). We denote the transcription of $\Phi^{\text{E}}_k \in \varPhi$ to the new variables as $(J_k,\nu_k^\text{E}) \in \bar{\varPhi}$. 
Optimizing over $J_k$ and $\nu_k^\text{E}$ yields the convex program:
\begin{align} \label{eq:ConvOpt}
  \begin{split}
    \eta^* = &\min_{\nu^\text{E}_k,J_k} \ \eta(\hat{\nu}) = \sum_{k\in \timeHorz}\Big(r^\textBA_k - \gamma^\text{E}_k\Big)^2  \\
\text{s.t.} \quad &\nu^{\text{E}}_{k+1} = \mathbb{1}^TJ_kG^{\text{E}}_k, \quad \nu^\text{E}_{\timeHorz(0)} = \hat{\nu}, \quad \gamma^\text{E}_k = \nu_k^\text{E}C^\text{E}, \\ 
&\nu^{\text{E}}_k,J_k \in [0,1], \ (\nu^{\text{E}}_k)^T = J_k\mathbb{1}, \ (J_k,\nu_k^\text{E}) \in \bar{\varPhi}.
  \end{split}
\end{align}
Once the convex problem is solved, the grid support control policies need to be recovered from it by using the relation~\eqref{eq:changVar}. If the matrix $\text{diag}(\nu_k^\text{E})$ is invertible, then the policy can be  obtained trivially from inversion of $\text{diag}(\nu_k^\text{E})$. If $\text{diag}(\nu_k^\text{E})$ is not invertible, then slight care is required when reconstructing a policy from the solution of~\eqref{eq:ConvOpt}. We describe this in the following Lemma.

\begin{lemma} \label{lem:polAlg}
	Suppose for all $k\in \timeHorz$ that $\nu_k^\text{E}$ and $J_k$ satisfy the constraints in problem~\eqref{eq:ConvOpt}. Then, there exists matrices $H_k=H_k(\nu_k^\text{E})$ and $W_k=W_k(\nu_k^\text{E})$ so that for all $k\in \timeHorz$ the quantity $\Phi_k^\text{E} = H_kJ_k + W_k$  satisfies~\eqref{eq:changVar} and $\Phi_k^\text{E}\in \varPhi$.
\end{lemma}
\ifshowArxivAlt
\begin{proof}
	See Appendix~\ref{app:polConsProof}.
\end{proof}
\fi
\ifshowArxiv
\begin{proof}
	See Appendix~\ref{app:polConsProof}.
\end{proof}
\fi
Exact construction of $H_k$ and $W_k$ is given in the proof of Lemma~\ref{lem:polAlg}. \emph{Hence, the proof of Lemma~\ref{lem:polAlg} provides an algorithm for computing grid support control policies that are feasible for the problem~\eqref{eq:nonConvOpt} from the solutions of the convex problem~\eqref{eq:ConvOpt}.} Further the two problems have a certain equivalence described here in the following Theorem.

\begin{thmm}\label{thm:equivalance}
	Denote $\eta^*_{\text{CVX}}$ the optimal cost for~\eqref{eq:ConvOpt} and $\eta^*_{\text{NCVX}}$ the optimal cost for~\eqref{eq:nonConvOpt} we have that $\eta^*_{\text{CVX}} = \eta^*_{\text{NCVX}}$.  
\end{thmm}
\ifshowArxivAlt
\begin{proof}
	See Appendix in~\cite{CoffmanUnifiedArxiv:2021}.
\end{proof}
\fi
\ifshowArxiv
\begin{proof}
	See Appendix~\ref{app:equivOptProof}.
\end{proof}
\fi 
This result, for a similar problem setup, is also reported in~\cite{BenenatiTractableCDC:2019}.
While we have no guarantee on the difference of the argument minimizers (and hence the policies obtained from both), Theorem~\eqref{thm:equivalance} says that the policies will produce the same tracking performance. Further, from Lemma~\ref{lem:polAlg}, the policies produced from either problem are guaranteed to ensure TCL QoS.

\subsubsection{Computational considerations}
The dimension of the program~\eqref{eq:ConvOpt} can be quite large, so that even though it is convex obtaining a solution requires care. We discuss now some practical considerations that we found necessary to consider when solving the problem~\eqref{eq:ConvOpt}. 

Due to the structure of $\Phi^{\text{E}}_k$, we do not need to declare every element in the matrix $J_k$ as a decision variable since many of these elements will be zero. For instance, we see that $\text{diag}(\nu^\text{E}_k)\Phi^\text{E}_k$ is a block matrix, where further each matrix block is diagonal. We express this as: $\text{diag}(\nu^\text{E}_k)\Phi^\text{E}_k =$
\begin{align} \label{eq:blockPolMat}
&\begin{bmatrix} \nonumber
B_{\textOff,\textOff}[k] & B_{\textOff,\textOn}[k] & \mathbf{0} & \mathbf{0} \\
\mathbf{0} & \mathbf{0} & B_{\textOn,\textOff}[k] & B_{\textOn,\textOn}[k]
\end{bmatrix}	\\ &\delequalRHS \ \text{sparse}(J_k),
\end{align}
where, e.g., $B_{\textOff,\textOff}[k] = \text{diag}(\nu^\text{E}_{\textOff}[k])(I - \Phi^\text{E}_{\textOff}[k])$. The other diagonal matrices appearing in~\eqref{eq:blockPolMat} can be inferred by carrying out the matrix multiplication. 

If $J_k$ was declared directly as a decision variable the problem~\eqref{eq:ConvOpt} would have $(8N^2+2N(\tau+1)) T_{\text{plan}}$ primal variables, whereas the problem with $\text{sparse}(J_k)$ as a decision variable only has $2N T_{\text{plan}}(\tau+3)$ primal variables. As an example, consider $N=12$, $T_{\text{plan}} = 360$, and $\tau = 5$, which are values used in numerical results reported later. The problem~\eqref{eq:ConvOpt} without the structure exploited has $\approx 0.5$ million decision variables, but only $\approx 75,000$ when the structure is exploited.

We also have found it helpful to include constraints of the form,
\begin{align} \label{eq:constAddOne}
	&\phi^\textPolName_{\textOff}(\text{on} \ \vert \ j-1)\nu_\textOff[\lambda^{j-1},0,k] \leq \phi^\textPolName_{\textOff}(\text{on} \ \vert \ j)\nu_\textOff[\lambda^{j},0,k], \\ \label{eq:constAddTwo}
	&\phi^\textPolName_{\textOn}(\text{off} \ \vert \ j+1)\nu_\textOn[\lambda^{j+1},0,k] \leq \phi^\textPolName_{\textOn}(\text{off} \ \vert \ j)\nu_\textOn[\lambda^j,0,k],
\end{align}
so to suggest that the switching on (resp., switching off) probability increases as temperature increases (resp., decreases). Adding the constraints~\eqref{eq:constAddOne}-\eqref{eq:constAddTwo} to the problem~\eqref{eq:ConvOpt} is straightforward as both the LHS and RHS of the inequalities are elements in the matrix $J_k$.

Matlab implementation of~\eqref{eq:ConvOpt} and the algorithm to extract the policies from $J_k$ (described in the proof of Lemma~\ref{lem:polAlg}) is available at~\cite{CoffmanUnifiedCode:2021}.

\subsubsection{Communication burden} \label{sec:comBurd}
Once solved, the policies obtained from~\eqref{eq:ConvOpt} need to be sent to each individual TCL. Many of the policy state values are constrained to either zero or one, which could be pre-programmed into each TCL. At each time index, $q-2$ (for the on to off policy) plus $N-m-1$ (for the off to on policy) numbers are not constrained and need to be sent from the BA to each TCL. Recall that the numbers $m$ and $q$ are temperature bin indices (see Figure~\ref{fig:cvLayout}) and $N$ is the number of temperature bins. For illustrative purposes, consider the values used in numerical experiments reported in the sequel: $N=12$ with $q = 10$ and $m = 2$ and a time discretization $\Delta t = 1$  minute. Since $N=q+m$, then the BA has to broadcast $2(q-1) = 18$  numbers every 1 minute to the TCLs. Each TCL receives the same 18 numbers.

Communication from TCLs to the BA - about their temperature and on/off state - is needed at the beginning of every planning period so that the BA can determine the initial condition $\hat{\nu}$ in~\eqref{eq:ConvOpt}. The frequency of this feedback is a design choice. In our numerical simulations reported later, a planning horizon of 6 hours was used, and this feedback was necessary only once in six hours.More frequent loop closure may be needed for higher robustness to uncertainty in weather prediction etc., a topic outside the scope of this paper.

\section{Numerical experiments}\label{sec:numExp}
Simulation involving coordination of $\numTCLs = 20,000$ TCLs through our proposed framework is presented here. Recall the two parts of the coordination architecture shown in Figure~\ref{fig:controlArch}: (i) planning and (ii) real time control. Planning refers to the solution of the problem~\eqref{eq:ConvOpt} at the BA to compute the following two things for the planning period $\timeHorz$:
\begin{enumerate}
	\item $r_k$: the reference power consumption of the TCL collection, given the problem data $r^\textBA_k$. 
	\item $\phi^{\textPolName,*}_\textOff[k]$ and $\phi^{\textPolName,*}_\textOn[k]$: grid support control policies for each TCL.
        \end{enumerate}
        This computation is performed at $\timeHorz(0)$.  Real time control is then the implementation of the grid support policies by each TCL to make on/off decisions in real time. We imagine the BA broadcasts the policies $\phi^{\textPolName,*}_\textOff[k]$ and $\phi^{\textPolName,*}_\textOn[k]$ at each $k$, though it can also broadcast all the policies, for all $k \in \timeHorz$, at $\timeHorz(0)$ and not broadcast again until the beginning of the next planning horizon. 

        The goal of the numerical simulations of real time control is to show the following.
\begin{enumerate}
	\item When each TCL uses the policies $\phi^{\textPolName,*}_\textOff[k]$ and $\phi^{\textPolName,*}_\textOn[k]$ to decide on/off actuation, the collection's power demand indeed tracks $r_k$.
	\item Every TCL's QoS constraints - both temperature and cycling - are satisfied at all times.
\end{enumerate}
Temperature of each TCL is computed in these simulations with the ODE model~\eqref{eq:detModelTCL}. 

\begin{table}
	\centering
	\caption{Simulation Parameters}
	\setlength{\arrayrulewidth}{0.03cm}
	\begin{tabular}{ l |c| c | l |c| c}
          \hline 
		Par. & Unit & value & Par. & Unit & value \\ \hline
		\numTCLs & N/A & 2$\times 10^4$ & $\eta$ & $\frac{\text{kW-e}}{\text{kW-th.}}$ & $2.5$ \\\hline
		$C$ & kWh$/^{\circ}$C & 1 & $P_0$ & kW & 5.5  \\\hline
		$\lambda^{\textLB}$ & $^{\circ}$C & 20 & $\lambda^{\textUB}$ & $^{\circ}$C & $22$\\\hline
		$(\Delta t)\tau$ & Mins. & 5 & $P_{\textAgg}$ & MW & 110\\ \hline
		$R$ & $^{\circ}$C$/$kW & 2 & $\Delta t$ & Mins. & 1\\ \hline
		$q$ & N/A & 10 & $m$ & N/A & 2\\ \hline
		$N$ & N/A & 12 & $T_{\text{plan}}$ & N/A & 360 \\ \hline
	\end{tabular}
\label{tab:BP}
\end{table}

\subsection{Planning}
The demand needed for demand-supply imbalance at the BA, $r^\textBA_k$, is chosen arbitrarily, and shown in Figure~\ref{fig:optProbRes} (top). It is infeasible for the collection: sometimes negative and sometimes far higher than the maximum power demand of the collection. This is done to simulate a realistic scenario in which many sources of demand and generation, not just TCLs, are managed by the BA.  

The baseline demand trajectory is defined by the equation~\eqref{eq:powerAgg}, which is approximately the power consumption for this collection of air conditioners under thermostat control. The ambient air temperature is time varying and is obtained from  \url{wunderground.com} for a typical summer day in Gainesville, Florida, USA. The other parameters that affect the Markov model are shown in Table~\ref{tab:BP}.

Planning computations are done with Matlab and CVX~\cite{cvx} using a desktop Linux machine, with $N=12$, and for a six hour planning horizon with 1 minute discretization ($T_{\text{plan}} = 360$). The problem~\eqref{eq:ConvOpt} takes about a minute to solve. The quantity $r^\textBA_k$, the baseline power $\bar{P}_k$, and the reference signal $r_k$, obtained from solving~\eqref{eq:ConvOpt}, are shown in Figure~\ref{fig:optProbRes} (top).  The optimal reference for the collection,  $r_k$, is as close to $r^\textBA_k$ as the dynamics of TCLs allows without violating their QoS constraints; recall Remark~\ref{rem:tclCap}. Figure~\ref{fig:optProbRes} (bottom) shows the two grid support control polices for one time instant. 

\subsection{Real time control}
The power consumption of the collection making on/off decisions according to the obtained policies is shown in Figure~\ref{fig:refTrackCent} (top). The figure shows that the TCLs are able to collectively track the reference signal $r_k$. We emphasize that the computational effort at each TCL is negligible. Recall Remark~\ref{rem:impRandPol}: once a TCL receives a grid support policy ($\approx$ 18 floating point numbers, see Section~\ref{sec:comBurd}) it only has to measure its current state (temperature and on/off mode) and generate a uniformly distributed random number in $[0, 1]$ to implement the policy. 

Verification of the grid support policies in ensuring QoS is shown in Figure~\ref{fig:refTrackCent}. The bottom plots shows a histogram of the times between switches for 300 randomly chosen TCLs. The middle plot shows a histogram of temperature from 200 randomly chosen TCLs' temperature trajectories. The histograms show that the policies designed with~\eqref{eq:ConvOpt} indeed satisfy the QoS constraints, which is specified by the vertical lines in the figures. Some TCLs do escape the temperature deadband by a little bit, which is expected and occurs also in thermostatic control: the sensor must \emph{first} register a value outside the deadband in order decide to switch the on/off state.

\begin{figure}
	\centering
	\begin{minipage}{0.5\textwidth}
		\includegraphics[width = 1\columnwidth]{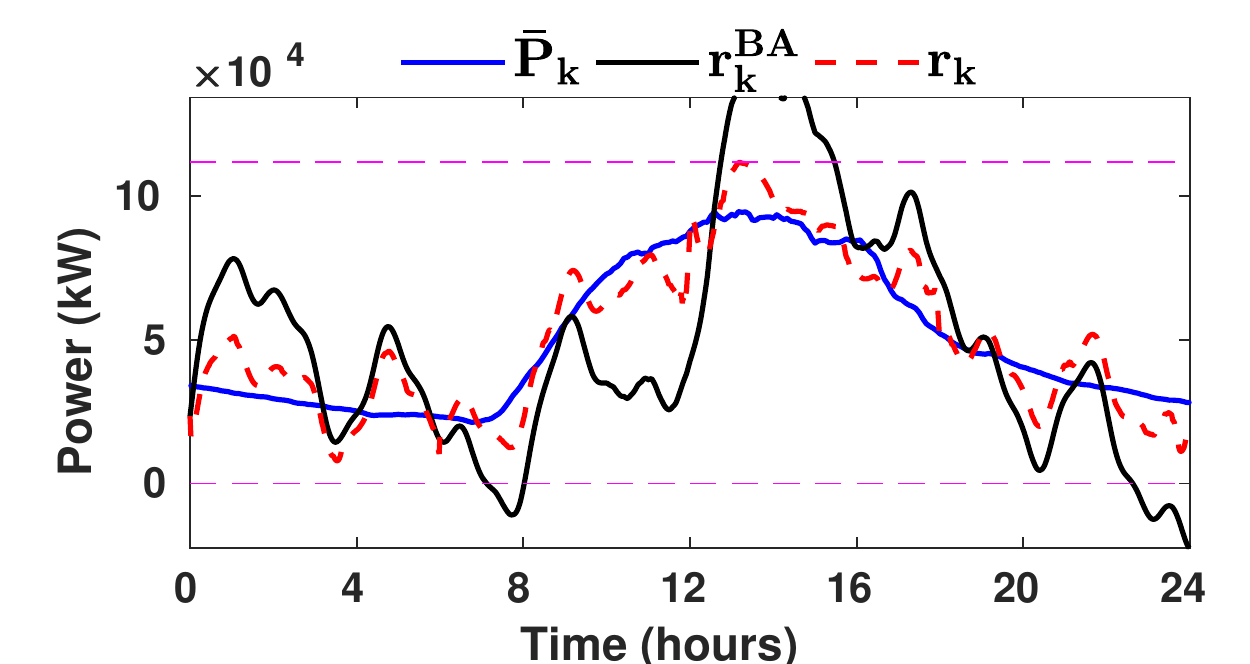}
	\end{minipage}
	\begin{minipage}{0.5\textwidth}
		\includegraphics[width = 1\columnwidth]{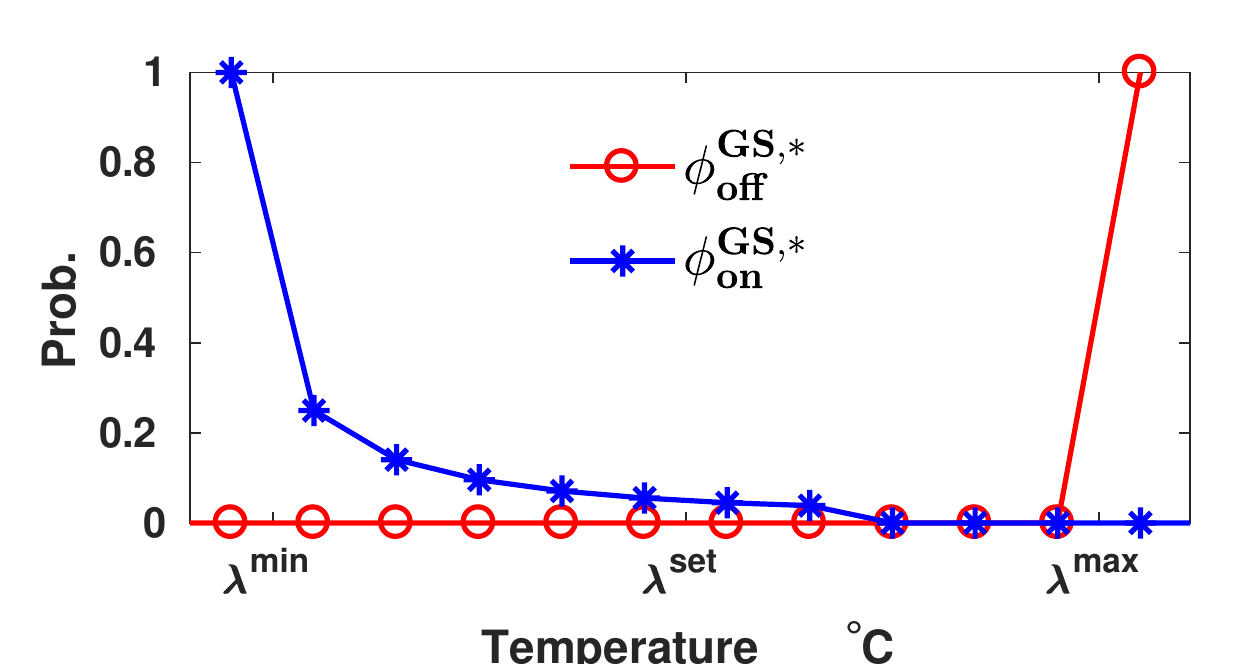}
	\end{minipage}
	\caption{(Top): The quantity $r_k$ obtained from solving~\eqref{eq:ConvOpt}, the dashed horizontal lines represent all of the TCLs on (top line) and off (bottom line). (Bottom): Grid support control policies, obtained from solving~\eqref{eq:ConvOpt}, at one time instance.}
	\label{fig:optProbRes}
\end{figure}

\begin{figure}[h]
	\centering
	\begin{minipage}{0.5\textwidth}
		\includegraphics[width = 1\columnwidth]{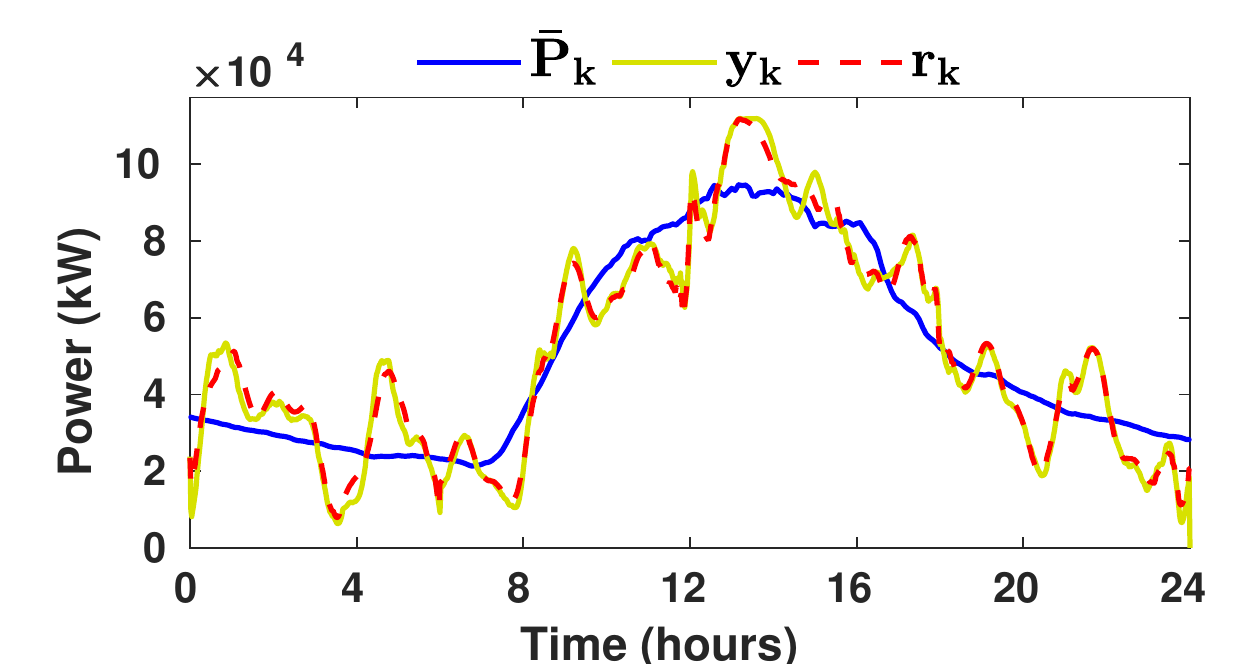}
	\end{minipage}
	\begin{minipage}{0.5\textwidth}
		\includegraphics[width = 1\columnwidth]{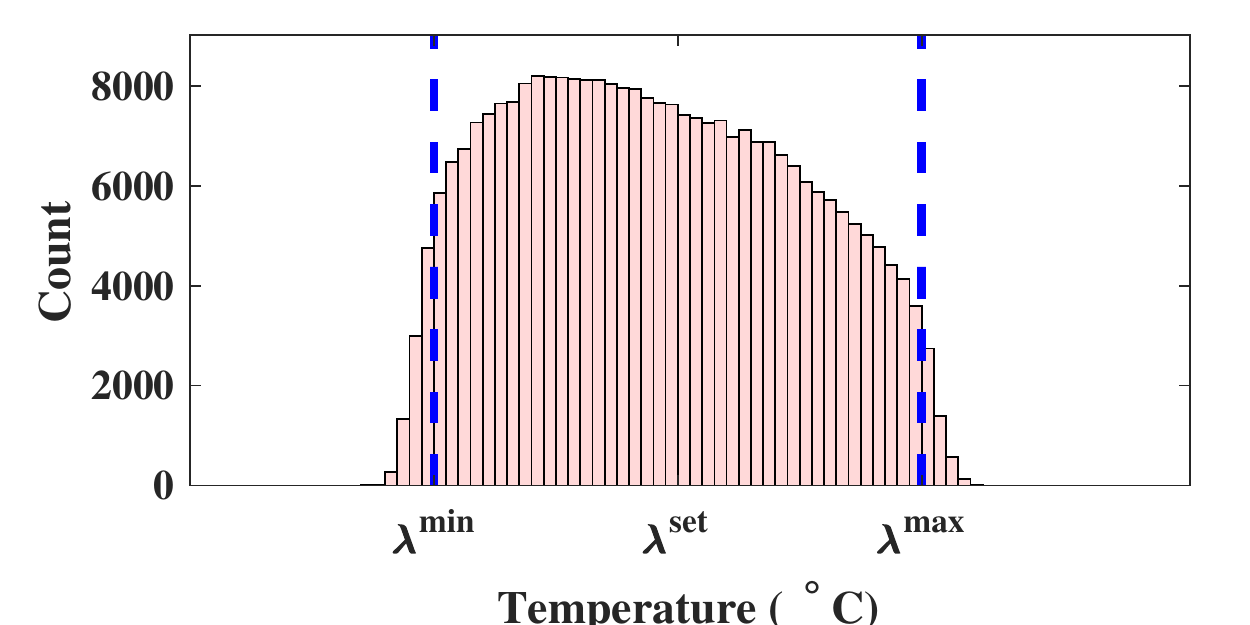}
	\end{minipage}
	\begin{minipage}{0.5\textwidth}
		\includegraphics[width = 1\columnwidth]{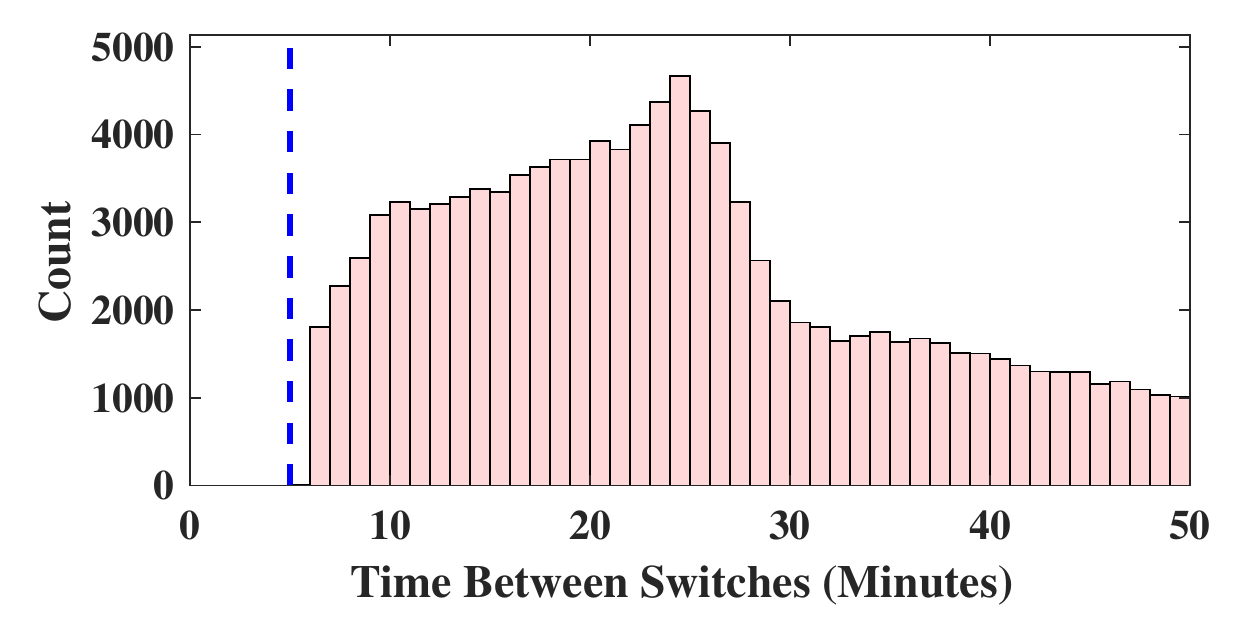}
	\end{minipage}
	\caption{(Top): Reference tracking results for the TCLs under the influence of the grid support control policies obtained by solving~\eqref{eq:ConvOpt}. (Middle): Histogram of the 200 TCL's temperature trajectories over the entire simulation horizon. (Bottom): Histogram of the time between switches over 3000 TCLs with the vertical line representing the minimum allowable time between switches.}
	\label{fig:refTrackCent}
\end{figure}

\ifx 0
Figure~\ref{fig:histTrackCent} shows the state of the TCL population predicted by the aggregate model 24 hours in advance, along with the histogram of the population obtained from simulating 20,000 TCLs. Each TCL in the simulation makes on/off decisions using the grid support policies mentioned above but the feedback from the TCLs to the BA is completely turned off. Despite this open loop nature, the BA is able to predict the histogram of the ensemble one day in advance. 

\begin{figure}
	\centering
	\begin{minipage}{0.5\textwidth}
		\includegraphics[width = 1\columnwidth]{histTrackingCentralized_withCycleState_notLocked_optPol.pdf}
	\end{minipage}
	\begin{minipage}{0.5\textwidth}
		\includegraphics[width = 1\columnwidth]{histTrackingCentralized_withCycleState_Locked_optPol.pdf}
	\end{minipage}
	\caption{(Top): Histogram of the ensemble for the devices that are on, not stuck, and under the influence of the grid support control policies obtained by solving~\eqref{eq:ConvOpt}. (Bottom): Histogram of the ensemble for the devices that are on, stuck, and under the influence of the grid support control policies obtained by solving~\eqref{eq:ConvOpt}. Green line is the model prediction, while the histograms are obtained from simulation data with real-time control with 20,000 air conditioners. }
	\label{fig:histTrackCent}
\end{figure}

\fi
\section{Conclusion} \label{sec:conc}

In this work we present a unified framework for the distributed control of TCLs. The framework enables: (i) reference planning for a collection of TCLs and (ii) design of a randomized control policy for the individual TCLs, so that both the BA's requirement and consumers' QoS are satisfied. The resulting framework is (i) scalable to an arbitrary number of loads and is implemented through \emph{local} feedback and minimal communication, (ii) able to guarantee both temperature and cycling constraints maintenance in each TCL, and (iii) based on convex optimization. Matlab/cvx implementation is publicly available~\cite{CoffmanUnifiedCode:2021}.

There are several avenues for future work. The optimal control problem is solved in an open-loop fashion here. Feedback from TCLs is used only to compute an initial condition that is needed as problem data for the off- line planning problem. It is straightforward to close the loop between the TCL collection and the BA with greater frequency for robustness to uncertainty in weather forecast and TCL parameters. It will be of interest to identify scenarios where closing loop, say, by using Model Predictive Control, is (i) necessary, and (ii) at what frequency should information be communicated from the TCLs to the BA. Another avenue is to investigate how the problem~\eqref{eq:ConvOpt} could be solved at each TCL, intermittently, instead of at the BA. Since the computational power of the processor at each TCL is lower than that of the processor at the BA, online distributed algorithms for convex optimization could play a role. The Fokker-Planck equations from~\cite{MalhameElectricTAC:1985} we used here are convenient for modeling TCL populations with a small deegree of heterogeneity. Distributed computation of optimal policies locally at each TCLs may help extend the method to a highly heterogeneous population of TCLs.

\ifshowArxivAlt
\bibliographystyle{plain}
\bibliography{\DiCEbibPATH/Barooah,\DiCEbibPATH/optimization,\DiCEbibPATH/grid,\DiCEbibPATH/ControlTheory,\DiCEbibPATH/distributed_control,\DiCEbibPATH/basics}
\fi

\ifshowArxiv

\fi

\appendix
\section{Proofs} \label{app:proof}
\renewcommand{\theequation}{\thesection.\arabic{equation}}

\ifshowArxiv
\subsection{Proof of Lemma 1} ~\label{app:rateMatProof}
See Appendix~\ref{app:pdeDisc} before reading this proof. Property (ii) is a consequence of the upwind difference scheme used. We see that for the internal CVs we have
\begin{align}
\text{off CVs:}& \quad -\Big(F^{i,+}_\textOff + D\Big) \\
\text{on CVs:}& \quad \Big(F^{i,-}_\textOn - D\Big). 
\end{align}
From Assumption \textbf{A.2} we have that $F^{i,-}_\textOn \leq 0$ and $F^{i,-}_\textOff \geq 0$ so that both of the above terms are negative. The upwind scheme is what ensured appropriate sign was added to the terms $F^{i,-}_\textOn$ and $F^{i,-}_\textOff$ so that the above coefficients are negative. Similar arguments can be applied for the off diagonal terms of the internal CVs and the boundary CVs.

To show property (i) we consider solely an internal CV for the off state as the arguments for all other CVs are identical in structure. Note that showing $A(t)\mathbb{1} = \mathbf{0}$ is equivalent to $\mathbb{1}^T\mathcal{A}(t) = \mathbf{0}^T$. Hence we need to show, for an arbitrary $i$ that all coefficients acting on $\nu_{\textOff}(\lambda^i,t)$ sum to $0$. We collect the coefficients corresponding to $\nu_{\textOff}(\lambda^i,t)$:
\begin{align} \nonumber
\text{From CV($i$)}:& \quad -F^{i,+}_\textOff(t) - D. \\ \nonumber
\text{From CV($i-1$)}:& \quad \frac{D}{2}.\quad \\ \nonumber
\text{From CV($i+1$)}:& \quad \frac{D}{2} + F^{i+1,-}_\textOff(t).\quad 
\end{align}
We then require the sum of these coefficients to be zero for all $t$ and any index $i$ for the internal off CVs, adding yields
\begin{align} \nonumber
	F^{i+1,-}_\textOff(t) - F^{i,+}_\textOff(t) = \frac{f_\textOff(\lambda^{i+1,-},t) - f_\textOff(\lambda^{i,+},t)}{\Delta\lambda} = 0
\end{align}
since by construction $\lambda^{i+1,-} = \lambda^{i,+}$ for the off CV's. This procedure can be repeated for $\nu_{\textOff}(\lambda^i,t)$ with $i \in \{1,m,N\}$, i.e., the boundary CVs in the off state and all of the on CVs in a similar fashion.
\fi

\ifshowArxiv
\subsection{Proof of Lemma 3} ~\label{app:condFactProof}
	If $\alpha = (\Delta t)^{-1}$, the diagonal elements of $A_k$ with $\alpha$ in them will go to zero and the non diagonal elements will go to 1. These non-diagonal elements with value $1$ are the red dots in Figure~\ref{fig:sparPattern} and encapsulate the thermostat control law. Thus the construction of $\Phi^\textTS$ with the canonical basis vectors. Now, multiplying out the matrix we have,
\begin{align}
	\Phi^\textTS G_k = \begin{bmatrix}
		\big(I-\Phi^\textTS_{\textOff}\big)P_k^{\textOff} & \Phi^\textTS_{\textOff}S_k^{\textOff} \\
		\Phi^\textTS_{\textOn}S_k^{\textOn} & \big(I-\Phi^\textTS_{\textOn}\big)P_k^{\textOn}
	\end{bmatrix}
\end{align}
where $\big(I-\Phi^\textTS_{\textOff}\big)P_k^{\textOff}$ (respectively, $\big(I-\Phi^\textTS_{\textOn}\big)P_k^{\textOn}$) is the matrix $P_k^{\textOff}$ (respectively, $P_k^{\textOn}$) but with the last (respectively, first) row zeroed out. The exact opposite statement is true for $\Phi^\textTS_{\textOff}P_k^{\textOn}$ and $\Phi^\textTS_{\textOn}P_k^{\textOff}$. Hence, by definition of the matrices in $G_k$ we have $P_k = \Phi^\textTS G_k$ where each non-zero element holds the interpretation~\eqref{eq:condIndFact}.
\fi

\subsection{Proof of Lemma 4} \label{app:polConsProof}
We define the following transformation for $l\in\{0,\dots,\tau\}$ and $j\in\{1,\dots,N\}$ as
\begin{align} \label{eq:indTran}
T(j,l) = lN + j
\end{align} 
that maps the integers $j$ and $l$ that label the state values to the absolute index of either of the vectors $\nu^\text{E}_\textOff$ and $\nu^\text{E}_\textOn$.
Now consider the following two sets
\begin{align} \label{eq:WmathcalOff}
\mathcal{W}_\textOff &\triangleq \Big\{(u,j,l) \in \stateSet \ \Big\vert \ \phi^\text{E}_\textOff(u \ \vert \ j, \ l) = \beta_\textOff(u,j,l) \Big\} \\ \label{eq:WmathcalOn}
\mathcal{W}_\textOn &\triangleq \Big\{(u,j,l) \in \stateSet \ \Big\vert \ \phi^\text{E}_\textOn(u \ \vert \ j, \ l) = \beta_\textOn(u,j,l) \Big\}.
\end{align}  
The values $\beta_\textOff$ and $\beta_\textOn$ are chosen to ensure the structural requirements in~\eqref{eq:randPolOff2On}-\eqref{eq:randPolOn2Off} and~\eqref{eq:expPolStruct}. For example, for $l = 1$ and $u =\text{on}$ we have that $\beta_\textOff(\textOn,\cdot,1) = \phi^{\text{TS}}_\textOff(\textOn \vert \cdot)$ (and hence $\beta_\textOff(\textOff,\cdot,1) = 1-\phi^{\text{TS}}_\textOff(\textOn \vert \cdot)$) so to enforce the structural requirement in~\eqref{eq:expPolStruct}.
 Define for each $k \in \timeHorz$ , $u,v \in \{\textOn,\textOff\}$, $j \in \{1,\dots,N \}$, and $l \in \{ 0,\dots,\tau \}$ the following vectors
\begin{align} \nonumber
h^u_k[T(j,l)] &\triangleq \begin{cases}
(\nu_u^\text{E}[\lambda^j,l,k])^{-1} & \text{if} \ \nu_u^\text{E}[\lambda^j,l,k] > 0. \\
0 & \text{otherwise}.
\end{cases} \\ \nonumber
w^{u,v}_k[T(j,l)] &\triangleq \begin{cases}
\beta_v(u,j,l)  &\text{if} \ (u,j,l) \in \mathcal{W}_v \ \text{and} \  \\ &\nu_v^\text{E}[\lambda^j,l,k] = 0. \\
0.5  &\text{if} \ (u,j,l) \notin \mathcal{W}_v \ \text{and} \  \\ &\nu_v^\text{E}[\lambda^j,l,k] = 0. \\
0 & \text{otherwise}.
\end{cases}
\end{align}
where $T(\cdot,\cdot)$ is defined in~\eqref{eq:indTran}, $\mathcal{W}_\textOff$ in~\eqref{eq:WmathcalOff}, and $\mathcal{W}_\textOn$ in~\eqref{eq:WmathcalOn}. Let $W^{u,v}_k = \text{diag}(w^{u,v}_k)$ and $H^{u,v}_k = \text{diag}(h^{u,v}_k)$, and construct the following matrices
\begin{align}
H_k &= \begin{bmatrix}
H_k^\textOff & \mathbf{0} \\
\mathbf{0} & H_k^\textOn
\end{bmatrix}, \quad \text{and} \\
W_k &= \begin{bmatrix}
W^{\textOff,\textOff}_k & W^{\textOff,\textOn}_k & \mathbf{0} & \mathbf{0} \\
\mathbf{0} & \mathbf{0} & W^{\textOn,\textOff}_k & W^{\textOn,\textOn}_k
\end{bmatrix}.
\end{align}

We first show that $\Phi_k^\text{E} = H_kJ_k + W_k$ satisfies~\eqref{eq:changVar}. Note that $\text{diag}(\nu_k^\text{E})W_k = \mathbf{0}$ since by construction if the $i^{th}$ row of $W_k$ has a non zero entry then the $i^{th}$ diagonal entry of $\text{diag}(\nu_k^\text{E})$ is zero. The product $\text{diag}(\nu_k^\text{E})H_k$ is a diagonal matrix with with entries of either zero or one. The zero entries also correspond to the zero entries of $\nu_k^\text{E}$. In this case, the respective entry in $J_k$ is also zero so that $\text{diag}(\nu_k^\text{E})H_kJ_k = J_k$ as desired.

We now show that $\phi_k^\text{E} = H_kJ_k + W_k \in \varPhi$. First consider an arbitrary state indexed by $(\textOff, j, l)$ at time $k$, if the corresponding value in $\nu_{\textOff}[\lambda^j,l,k] > 0$ then the two policy values are defined as
\begin{align}
		&\frac{\Prob\left(m_{k+1} = \text{on},\ I_{k} = j,\ \Locked_k = l,  \ m_k = \text{off}\right)}{\nu_{\textOff}[\lambda^j,l,k]} \\
		&\frac{\Prob\left(m_{k+1} = \text{off},\ I_{k} = j,\ \Locked_k = l,  \ m_k = \text{off}\right)}{\nu_{\textOff}[\lambda^j,l,k]}.
\end{align}
If either of the above values are fixed in the constraint set $\Phi$, then the constraint~\eqref{eq:polCons} will ensure this. Further, since we have that $\nu_k^\text{E},J_k \in [0,1]$ and that $(\nu_k^\text{E})^T = J_k\mathbb{1}$ this ensures that the above policy values are within $[0,1]$ and sum to 1. The above argument is valid for any pair of state values such that the corresponding value of $\nu_k^\text{E}$ is non-zero. If the corresponding value of $\nu_{\textOff}[\lambda^j,l,k] = 0$ and the policy (conditioned on this state value) has a constraint, the first if case in the definition of $w^{u,v}_k$ ensures this constraint. Further, the constraint values must also be chosen to ensure the respective policy values are in $[0,1]$ and sum to one. Lastly, if $\nu_{\textOff}[\lambda^j,l,k] = 0$ and there is no constraint for the policy conditioned on this state value the second if case in the definition of $w^{u,v}_k$ ensures the policy value sums to 1 and the respective elements are in $[0,1]$. Thus $\Phi_k^\text{E} \in \varPhi$ for all $k\in\timeHorz$.

\ifshowArxiv
\subsection{Proof of Theorem 1} \label{app:equivOptProof}
The proof structure is similar to the one in~\cite{BenenatiTractableCDC:2019}. The idea is to exploit the fact that: (i) $\nu_k^\text{E}$ is a decision variable for both optimization problems~\eqref{eq:ConvOpt} and~\eqref{eq:nonConvOpt} and (ii) the objective function is the same for both problems and solely a function of the marginal $\nu_k^\text{E}$. We rewrite these problem compactly below,
\begin{align}
	\eta^*_{\text{CVX}} &= \min_{(\nu^\text{E},J) \in X} \eta(\nu^\text{E}), \\
		\eta^*_{\text{NCVX}} &= \min_{(\nu^\text{E},\Phi^\text{E}) \in Y} \eta(\nu^\text{E}),
\end{align}
where the sets $X$ and $Y$ collect all of the relevant constraints for the problems. The variables $\nu^\text{E}$, $\Phi^\text{E}$, and $J$ are concatenated over the considered finite time horizon and hence are not sub-scripted by $k$. We proceed by showing that $\eta^*_{\text{CVX}} \leq \eta^*_{\text{NCVX}}$ and $\eta^*_{\text{NCVX}} \leq \eta^*_{\text{CVX}}$ to give the desired result.

\subsubsection{$\eta^*_{\text{CVX}} \leq \eta^*_{\text{NCVX}}$}
Pick any argument minimizer that achieves value $\eta^*_{\text{NCVX}}$ and denote the pair as $(\nu^\text{E}_{\text{NCVX}},\Phi^\text{E}_{\text{NCVX}})$. Trivially construct $J$ through the relation~\eqref{eq:changVar} so that this constructed $J$ and $\nu^\text{E}_{\text{NCVX}}$ (that is optimal for~\eqref{eq:nonConvOpt}) are also feasible for~\eqref{eq:ConvOpt}, i.e., $(\nu^\text{E}_{\text{NCVX}},J) \in X$. This is since $\mathbb{1}^T\text{diag}(\nu^{\text{E}}_k) = \nu^{\text{E}}_k$  and $\mathbb{1} = \Phi^{\text{E}}_k\mathbb{1}$. Hence we have that
\begin{align}
	\eta^*_{\text{CVX}} &= \min_{(\nu^\text{E},J) \in X} \eta(\nu^\text{E}) \leq \eta(\nu^\text{E}_{\text{NCVX}}) = \eta^*_{\text{NCVX}}
\end{align}   
since by definition $\eta^*_{\text{CVX}}$ is the minimum value over the set of feasible solutions.

\subsubsection{$\eta^*_{\text{NCVX}} \leq \eta^*_{\text{CVX}}$}
We take a pair $(\nu^\text{E}_{\text{CVX}},J_{\text{CVX}})$ that achieve optimal cost $\eta^*_{\text{CVX}}$ and construct a feasible solution for~\eqref{eq:nonConvOpt}, denoted $(\eta^\text{E}_{\text{NCVX}},\Phi^\text{E}_{\text{NCVX}})$, as follows (for each $k$)
\begin{align}
	\Phi^\text{E}_{k,\text{NCVX}} &= H_kJ_k + W_k, \quad \text{and} \\
	\nu^\text{E}_{k,\text{NCVX}} &= \nu^\text{E}_{k,\text{CVX}}.
\end{align}
Where $H_k$ and $W_k$ are defined in Lemma~\ref{lem:polAlg}. This constructed solution is then feasible for~\eqref{eq:nonConvOpt} as the constraint $\Phi^\text{E}_{\text{NCVX}} \in \varPhi$ is part of the result in Lemma~\ref{lem:polAlg} and
\begin{align}
	\nu^\text{E}_{k,\text{NCVX}}\Phi^\text{E}_{k,\text{NCVX}}G_k^\text{E} &= \nu^\text{E}_{k,\text{NCVX}}\big(H_kJ_k+W_k\big)G_k^\text{E} \\ &= \mathbb{1}^TJ_kG_k^\text{E} = \nu^\text{E}_{k+1,\text{NCVX}}.
\end{align} 
The fact that $\nu^\text{E}_{k,\text{NCVX}}\big(H_kJ_k+W_k\big) = \mathbb{1}^TJ_k$ is since $\nu^\text{E}_{k,\text{NCVX}}W_k = \mathbf{0}$ and $\nu^\text{E}_{k,\text{NCVX}}H_kJ_k = \mathbb{1}^TJ_k$. The matrix $W_k$ only has non zero entries for row indices where the index of the row vector $\nu^\text{E}_{k,\text{NCVX}}$ is zero so that the resulting product is the zero vector. The product $\nu^\text{E}_{k,\text{NCVX}}H_k$ is a vector of ones and zeros, specifically, if the $i^{th}$ element of this vector is zero then the entire $i^{th}$ column of the matrix $J_k$ will be the zero vector. Thus the equivalence between $\nu^\text{E}_{k,\text{NCVX}}\big(H_kJ_k+W_k\big)$ and $ \mathbb{1}^TJ_k$.  Since the constructed solution is feasible we have that
\begin{align}
\eta^*_{\text{NCVX}} &= \min_{(\nu^\text{E},\Phi^\text{E}) \in Y} \eta(\nu^\text{E}) \leq \eta(\nu^\text{E}_\text{CVX}) = \eta^*_\text{CVX}
\end{align}
since by definition $\eta^*_{\text{NCVX}}$ is the minimum value over the set of feasible solutions.
\fi

\ifshowArxiv
\section{PDE discretization} \label{app:pdeDisc}
We denote the $i^{th}$ CV as CV($i$) and further adopt the following notational simplifications,
\begin{align} \nonumber
\mu_{\textOff}(\lambda^i,t) \triangleq \mu_{\textOff}(\lambda^i_{\textOff},t), \quad \text{and} \quad \mu_{\textOn}(\lambda^i,t) \triangleq \mu_{\textOn}(\lambda^i_{\textOn},t).
\end{align}
Highlighted red in Figure~\ref{fig:cvLayout} are the two control volumes to assist in enforcing boundary conditions that coincide with the thermostat policy~\eqref{eq:thermoContLaw}. This is discussed further in Appendix~\ref{app:boundCondProof} when the boundary conditions CVs are discretized.

\subsection{Internal CV's} \label{app:pdeDiscInt}
Consider the RHS of the pde~\eqref{eq:pdeOnMode} integrated over CV($i$):
\begin{align} \nonumber
&\int_{\text{CV(i)}}\bigg(\frac{\sigma^2}{2}\frac{\partial^2}{\partial \lambda^2}\big(\mu_{\textOn}(\lambda,t)\big)-\frac{\partial}{\partial \lambda}\big(f_{\textOn}(\lambda,t)\mu_{\textOn}(\lambda,t)\big)\bigg)d\lambda \\
\label{eq:arbOnIntCV}
&=\bigg(\frac{\sigma^2}{2}\frac{\partial}{\partial \lambda}\mu_{\textOn}(\lambda,t) -f_{\textOn}(\lambda,t)\mu_{\textOn}(\lambda,t) \bigg)\bigg\vert_{\lambda^{i,-}}^{\lambda^{i,+}},
\end{align}
where equality is by the divergence theorem~\cite{VersteegIntroductionBook:2007}.
Note, the points $\lambda^{i,-}$ and $\lambda^{i,+}$ are not control volume variables, but rather the boundaries of a single control volume. Hence, quantities in~\eqref{eq:arbOnIntCV} need to be approximated in terms of the nodal points of the neighboring control volumes. The approximations for the partial derivative are,
\begin{align} \label{eq:centDiff}
\frac{\partial}{\partial \lambda}\mu_{\textOn}(\lambda^{i,+},t) &\approx \frac{\mu_{\textOn}(\lambda^{i+1},t) - \mu_{\textOn}(\lambda^i,t)}{\Delta \lambda}, \quad \text{and} \\
\frac{\partial}{\partial \lambda}\mu_{\textOn}(\lambda^{i,-},t) &\approx \frac{\mu_{\textOn}(\lambda^{i},t) - \mu_{\textOn}(\lambda^{i-1},t)}{\Delta \lambda}.
\end{align}
For the integrated convective term, we use the so-called upwind scheme~\cite{VersteegIntroductionBook:2007}. This scheme elects the FVM equivalent of a forward or backward difference based on the sign of the convective velocity $f_{\text{on}}(\lambda,t)$. By assumption \textbf{A.2}, $f_{\textOn}(\lambda,t) \leq0$ and the upwind scheme prescribes:  
\begin{align} \nonumber
f_{\textOn}(\lambda^{i,-},t)\mu_{\textOn}(\lambda^{i,-},t) &= f_{\textOn}(\lambda^{i,-},t)\mu_{\textOn}(\lambda^{i},t), \quad \text{and}\\
f_{\textOn}(\lambda^{i,+},t)\mu_{\textOn}(\lambda^{i,+},t) &= f_{\textOn}(\lambda^{i,+},t)\mu_{\textOn}(\lambda^{i+1},t).
\end{align}
When the TCL is off $f_{\textOff}(\lambda,t)\geq 0$ (also by Assumption \textbf{A.2}) the upwind scheme prescribes:  
\begin{align} \nonumber
f_{\textOff}(\lambda^{i,-},t)\mu_{\textOff}(\lambda^{i,-},t) &= f_{\textOff}(\lambda^{i,-},t)\mu_{\textOff}(\lambda^{i-1},t), \ \text{and} \\
f_{\textOff}(\lambda^{i,+},t)\mu_{\textOff}(\lambda^{i,+},t) &= f_{\textOff}(\lambda^{i,+},t)\mu_{\textOff}(\lambda^{i},t).
\end{align}
Now returning to the discretization of the PDE~\eqref{eq:pdeOnMode} over an arbitrary internal CV. We approximate the LHS of~\eqref{eq:pdeOnMode} integrated over the control volume as,
\begin{align} \nonumber
\int_{\text{CV(i)}}\frac{\partial}{\partial t}\mu_{\textOn}(\lambda,t)d\lambda \approx \frac{d}{dt}\mu_{\textOn}(\lambda^i,t)\Delta\lambda = \frac{d}{dt}\nu_{\textOn}(\lambda^i,t),
\end{align}
where we have defined
\begin{align}
\nu_{\textOn}(\lambda^i,t)\triangleq \mu_{\textOn}(\lambda^i,t)\Delta\lambda.	
\end{align}
Now, denote the following
\begin{align}
D \triangleq \frac{\sigma^2}{(\Delta\lambda)^2}, \quad \text{and} \quad F_{\textOn}^i(t) \triangleq \frac{f_{\textOn}(\lambda^{i},t)}{\Delta\lambda},
\end{align}
where the quantities $F_{\textOff}^i(t)$, $F_{\textOn}^{i,+}(t)/F_{\textOff}^{i,+}(t)$, and $F_{\textOn}^{i,-}(t)/F_{\textOff}^{i,-}(t)$ are defined similarly to $F_{\textOn}^i(t)$, e.g., $F_{\textOff}^{i,+}(t) \triangleq f_\textOff(\lambda^{i,+},t)/\Delta\lambda$. Now equating the approximation of the RHS~\eqref{eq:pdeOnMode} with the approximation of the LHS of~\eqref{eq:pdeOnMode} we have,
\begin{align} \nonumber
\frac{d}{d t}\nu_{\textOn}(\lambda^i,t) &= \Big(F_{\textOn}^{i,-}(t) - D\Big)\nu_{\textOn}(\lambda^{i},t)+ \frac{D}{2}\nu_{\textOn}(\lambda^{i-1},t)  \\ \label{eq:stanOnCV}
&+\Big(\frac{D}{2}-F_{\textOn}^{i,+}(t)\Big)\nu_{\textOn}(\lambda^{i+1},t). 
\end{align}
The spatial discretization for the PDE~\eqref{eq:pdeOffMode} is similar and yields,
\begin{align} \nonumber
\frac{d}{d t}\nu_{\textOff}(\lambda^i,t) &=  \frac{D}{2}\nu_{\textOff}(\lambda^{i+1},t)  -\Big(F_{\textOff}^{i,+}(t) + D\Big)\nu_{\textOff}(\lambda^{i},t)\\ \label{eq:stanOffCV}
&+\Big(\frac{D}{2}+F_{\textOff}^{i,-}(t)\Big)\nu_{\textOff}(\lambda^{i-1},t),
\end{align} 
where $\nu_{\textOff}(\lambda^i,t)\triangleq \mu_{\textOff}(\lambda^i,t)\Delta\lambda.$

\subsection{Boundary CV's} \label{app:boundCondProof}
The boundary CVs are the CVs associated with the nodal values: $\lambda_{\textOn}^1$, $\lambda_{\textOn}^q$, $\lambda_{\textOn}^N$, $\lambda_{\textOff}^1$, $\lambda_{\textOff}^m$, and $\lambda_{\textOff}^N$. The superscript, for example the integer $q$ in $\lambda_{\textOn}^q$ represents the CV index. All boundary CVs can be seen in Figure~\ref{fig:cvLayout}. Discretization of the boundary CVs requires care for atleast two reasons. First, this is typically where one introduces the BCs of the PDE into the numerical approximation. Secondly, on finite domains the endpoints present challenges, for example, there is no variable $\mu_{\textOn}(\lambda^{N+1},t)$ for computation of the derivative values for node $\lambda^{N}_\textOn$.

The BC's for the coupled PDEs~\eqref{eq:pdeOnMode}-\eqref{eq:pdeOffMode} are~\cite{MalhameElectricTAC:1985}:
\begin{align} \nonumber
&\text{\bf{Absorbing Boundaries}:} \\ \label{eq:absorbBoundary}
&\qquad \qquad \qquad \mu_{\textOn}(\lambda^{\text{min}},t) = \mu_{\textOff}(\lambda^{\text{max}},t) = 0. \\ \nonumber
&\text{\bf{Conditions at Infinity}:} \\ \label{eq:contInfBC}
&\qquad \qquad \qquad \mu_{\textOn}(+\infty,t) = \mu_{\textOff}(-\infty,t) = 0. \\ \nonumber
&\text{\bf{Conservation of Probability}:} \\ \label{eq:consProbOne}
&\frac{\partial}{\partial\lambda}\bigg[\mu_{\textOn}(\lambda^{q,-},t)-\mu_{\textOn}(\lambda^{q-1,+},t)  -\mu_{\textOff}(\lambda^{N-1,+},t)\bigg]= 0. \\ \label{eq:consProbTwo}
&\frac{\partial}{\partial\lambda}\bigg[\mu_{\textOff}(\lambda^{m,+},t)-\mu_{\textOn}(\lambda^{2,-},t) -\mu_{\textOff}(\lambda^{m+1,-},t)\bigg]= 0. \\
&\nonumber \text{\bf{Continuity}:} \\ \label{eq:contBC_1}
&\qquad \qquad \qquad \mu_{\textOn}(\lambda^{q,-},t) = \mu_{\textOn}(\lambda^{q-1,+},t). \\  \label{eq:contBC_2}
&\qquad \qquad \qquad \mu_{\textOff}(\lambda^{m,+},t) = \mu_{\textOff}(\lambda^{m+1,-},t). 
\end{align}
As we will see, implementation of some of the above conditions will require a bit of care. However, some are quite trivial to enforce. For example, by default, the continuity conditions~\eqref{eq:contBC_1} and~\eqref{eq:contBC_2} are satisfied due to our choice of CV structure, since, for example, for any $i$ we have $\lambda_{\textOff}^{i,-}=\lambda_{\textOff}^{i-1,+}$ and $\lambda_{\textOff}^{i,+}=\lambda_{\textOff}^{i+1,-}$. 
\ifx 0
Further, the choice of grid also reduce the conservation of probability to the:
\begin{align} \label{eq:modConsProb}
\frac{\partial}{\partial \lambda}\mu_{\textOff}(\lambda^{N-1,+},t) = \frac{\partial}{\partial \lambda}\mu_{\textOff}(\lambda^{2,-},t) = 0,
\end{align}
which with the BC~\eqref{eq:absorbBoundary} specifies a periodic boundary condition. 
We can reduce~\eqref{eq:consProbOne} and~\eqref{eq:consProbOne} to~\eqref{eq:modConsProb} since the FVM, by design, conserves.Although, we return to this condition after dealing with the simpler conditions first.   
\fi

Now focusing on the conditions at infinity BC~\eqref{eq:contInfBC}, we enforce instead the following conditions:
\begin{align}
\frac{\partial}{\partial \lambda}\mu_{\textOff}(\lambda^{1,-},t) = 0, \quad \text{and} \quad
\frac{\partial}{\partial \lambda}\mu_{\textOn}(\lambda^{N,+},t) = 0.
\end{align}
Our computational domain cannot extend to infinity, where the BC~\eqref{eq:contInfBC} is required to hold, but the temperature values $\lambda_{\textOff}^1$ and $\lambda_{\textOn}^N$ are quite far away from the deadband and so the density here will be near zero.  

Now, consider the spatial discretization of the CVs associated with the BC at infinity. First considering the CV associated with the temperature $\lambda^1_{\textOff}$, we have that the differential equation is
\begin{align} \label{eq:BC_lambdaOff1}
\frac{d}{d t}\nu_{\textOff}(\lambda^1,t) &= \Big(-F_{\textOff}^{1,+}(t) - \frac{D}{2}\Big)\nu_{\textOff}(\lambda^{1},t)  \\  \nonumber
&+\Big(\frac{D}{2}+F_{\textOff}^{2,-}(t)\Big)\nu_{\textOff}(\lambda^{2},t).
\end{align}
Considering the CV associated with the temperature $\lambda^N_{\text{on}}$, we have
\begin{align} \label{eq:BC_lambdaOnN}
\frac{d}{d t}\nu_{\textOn}(\lambda^N,t) &= \Big(F_{\textOn}^{N,+}(t) - \frac{D}{2}\Big)\nu_{\textOn}(\lambda^{N},t)   \\ \nonumber
&+\Big(\frac{D}{2}-F_{\textOn}^{N,+}(t)\Big)\nu_{\textOn}(\lambda^{N-1},t).
\end{align}
In the above we make the assumption that $\nu_{\textOff}(\lambda^{1,-} - \Delta\lambda,t) = 0$ and $\nu_{\textOn}(\lambda^{N,+} + \Delta\lambda,t) = 0$. 

Now focus on the absorbing boundary~\eqref{eq:absorbBoundary} and conservation of probability \eqref{eq:consProbOne}-\eqref{eq:consProbTwo} boundary conditions. These BCs have the following meaning. The condition~\eqref{eq:absorbBoundary} clamps the density at the end of the deadband to zero. BC~\eqref{eq:consProbOne} reads: the net-flux across the temperature value $\lambda^q_{\textOn}$ is equal to the flux of density going from off to on. In order to enforce both~\eqref{eq:consProbOne} and \eqref{eq:consProbTwo} we will model the flux of density due to the thermostat control policy as a source/sink. Before doing this, we mention some issues with enforcing the BC~\eqref{eq:absorbBoundary}.

A TCL's temperature trajectory will not satisfy the BC~\eqref{eq:absorbBoundary} since to switch its mode the TCL's temperature sensor will have to register a value outside the deadband. Therfore, we introduce two additional CV's associated with the temperatures $\lambda^1_\textOn$ and $\lambda^N_\textOff$, which are shown in red in Figure~\ref{fig:cvLayout}. We then transfer the BC~\eqref{eq:absorbBoundary} to one on the added CVs, which becomes:
\begin{align} \label{eq:modfAbsCond}
\mu_{\textOn}(\lambda^{1,-},t) = \mu_{\textOff}(\lambda^{N,+},t) &= 0.
\end{align}
As mentioned, to enforce the conservation of probability BC we use a source/sink type argument, which we also enforce on the added CVs. To see what we mean by source/sink argument, consider the following: some rate of density is transferred out of the CV $\lambda^N_\textOff$ and into the CV $\lambda^q_\textOn$ (as depicted in Figure~\ref{fig:cvLayout}) due to thermostatic control. We model the sink as simply $-\nu_{\textOff}(\lambda^{N},t)$. The rate of the sink is then given as $-\gamma\nu_{\textOff}(\lambda^{N},t)$, where $\gamma>0$ is a modeling choice and a constant of appropriate units that describes the discharge rate. We shortly give insight on how to select a value for $\gamma$. Now discretizing the CV corresponding to the nodal value $\lambda^N_{\textOff}$ subject to the BC~\eqref{eq:modfAbsCond} and the sink $-\nu_{\textOff}(\lambda^{N},t)$ we obtain,
\begin{align} \label{eq:BC_lambdaOffN}
\frac{d}{dt}\nu_{\textOff}(\lambda^{N},t) &= \Big(\frac{D}{2}+F_{\textOff}^{N,-}(t)\Big)\nu_{\textOff}(\lambda^{N-1},t) \\ \nonumber &-\alpha\nu_{\textOff}(\lambda^{N},t), 
\end{align}
where $\alpha \triangleq \big(\gamma+ D\big)$. In obtaining the above, we have made the reasonable assumption that $\nu_{\textOff}(\lambda^{N,+}+\Delta\lambda,t) = 0$.  The quantity $\alpha\nu_{\textOff}(\lambda^{N},t)$ represents the rate of change of density from the CV $\lambda^N_{\textOff}$ to the CV $\lambda^q_{\textOn}$, as depicted in Figure~\ref{fig:cvLayout}. Consequently, to conserve probability, we must add this quantity as a source to the ode for the CV $\lambda^q_{\textOn}$, i.e., 
\begin{align}\label{eq:boundOff2On}
\frac{d}{dt}\nu_{\textOn}(\lambda^{q},t) &= \dots + \alpha\nu_{\textOff}(\lambda^{N},t). 
\end{align}
The dots in equation~\eqref{eq:boundOff2On} represent the portion of the dynamics for the standard internal CV (i.e., the RHS of~\eqref{eq:stanOnCV}) for the temperature node $\lambda^q_{\textOn}$. A similar argument is used for the BC~\eqref{eq:consProbTwo} with the CV's $\lambda^{1}_{\textOn}$ and $\lambda^{m}_{\textOff}$, and the corresponding differential equations are,
\begin{align} \label{eq:BC_lambdaOn1}
\frac{d}{dt}\nu_{\textOn}(\lambda^{1},t) &= \Big(\frac{D}{2}-F_{\textOn}^{1,+}(t)\Big)\nu_{\textOn}(\lambda^{2},t) - \alpha\nu_{\textOn}(\lambda^{1},t), \\ \label{eq:boundOn2Off}
\frac{d}{dt}\nu_{\textOff}(\lambda^{m},t) &= \dots + \alpha\nu_{\textOn}(\lambda^{1},t). 
\end{align}
To better understand the role of $\gamma$ consider the following example. Electing $\gamma$ in the above so that $\alpha = (\Delta t)^{-1}$, where $\Delta t$ is a time increment, has the following interpretation: all mass starting in state $\nu_{\textOff}(\lambda^{N},\cdot)$ at time $t$ is transferred out by time $t+\Delta t$ into the state $\nu_{\textOn}(\lambda^{q},\cdot)$. 

\subsubsection{Additional conditions}

Two additional conditions are enforced, namely that once mass is transferred to the nodes $\lambda^N_\textOff$ or $\lambda^1_\textOn$ it cannot ``travel backwards.'' For example, mass is transferred from $\lambda^N_\textOff$ entirely to the corresponding on temperature bin and no mass is transferred backwards to $\lambda^{N-1}_\textOff$. This corresponds to setting: (i) the coefficient on $\nu_{\textOff}(\lambda^N,t)$ in the ode for $\nu_{\textOff}(\lambda^{N-1},t)$ to zero and (ii) the coefficient on $\nu_{\textOn}(\lambda^1,t)$ in the ode for $\nu_{\textOn}(\lambda^{2},t)$ to zero.

\subsection{Overall system}

Now, combining the odes--\eqref{eq:stanOnCV} and~\eqref{eq:stanOffCV} for all of the internal CVs and \eqref{eq:BC_lambdaOff1}, \eqref{eq:BC_lambdaOnN}, \eqref{eq:BC_lambdaOffN}, \eqref{eq:boundOff2On}, \eqref{eq:BC_lambdaOn1}, \eqref{eq:boundOn2Off} for the BC CVs--we obtain the linear time varying system,
\begin{align} \label{eq:dynContMC-appendix}
\frac{d}{dt}\nu(t) = \nu(t)A(t).
\end{align}
\fi

\end{document}